\documentclass[11pt]{amsart}
\usepackage{amsmath,amsthm,amssymb,bm,bbm,dsfont}
\usepackage{amsaddr}
\usepackage{upgreek}
\usepackage[left=2cm,right=2cm,top=2cm,bottom=2cm]{geometry}
\usepackage{setspace}
\usepackage{graphicx}
\usepackage{verbatim}
\usepackage{float}
\restylefloat{table}
\usepackage{placeins}
\usepackage{array}
\usepackage{booktabs}
\usepackage{threeparttable}
\usepackage[update,prepend]{epstopdf}
\usepackage{multirow}
\usepackage{amsfonts,amssymb,dsfont}
\usepackage{mathrsfs}  
\usepackage[abs]{overpic}

\usepackage[usenames,dvipsnames]{color}
\usepackage[hidelinks]{hyperref}
\hypersetup{
	unicode=false,          
	pdftoolbar=true,        
	pdfmenubar=true,        
	pdffitwindow=false,     
	pdfstartview={FitH},    
	pdftitle={My title},    
	pdfauthor={Author},     
	pdfsubject={Subject},   
	pdfcreator={Creator},   
	pdfproducer={Producer}, 
	pdfkeywords={keyword1} {key2} {key3}, 
	pdfnewwindow=true,      
	colorlinks=true,        
	linkcolor=Red,          
	citecolor=ForestGreen,  
	filecolor=Magenta,      
	urlcolor=BlueViolet,    
}
\usepackage{doi}
\usepackage{url}
\usepackage{caption, subcaption}
\usepackage{enumitem}

\makeatletter
\ifx\@NODS\undefined%

\let\mathbb=\mathds
\else%
\fi
\makeatother

\usepackage[backend=bibtex,
hyperref=true,
url=true,
isbn=true,
backref=false,
style=ieee,
sorting=none,
block=none]{biblatex}
\bibliography{reference.bib}

\DeclareMathOperator{\Tr}{Tr}
\DeclareMathOperator{\tr}{\overline{tr}}

\DeclareMathOperator{\Ent}{Ent}

\DeclareMathOperator{\Var}{Var}

\newcommand{\be}{{\mathbf e}}

\def\0{{\mathbf{0}}}
\def\1{{\mathbf{1}}}
\def\2{{\mathbf{2}}}
\def\3{{\mathbf{3}}}
\def\4{{\mathbf{4}}}
\def\5{{\mathbf{5}}}
\def\6{{\mathbf{6}}}

\def\7{{\mathbf{7}}}
\def\8{{\mathbf{8}}}
\def\9{{\mathbf{9}}}

\def\bbM{{\mathbb{M}}}

\def\bbR{\mathbb{R}}

\def\be{\begin{equation}}
\def\ee{\end{equation}}
\def\bea{\begin{eqnarray}}
\def\eea{\end{eqnarray}}

\theoremstyle{plain}
\newtheorem{theo}{Theorem} 
\newtheorem{prop}[theo]{Proposition} 
\newtheorem{coro}[theo]{Corollary} 

\theoremstyle{definition}
\newtheorem{defn}[theo]{Definition} 

\theoremstyle{remark}
\newtheorem{remark}{Remark}[section]

\numberwithin{equation}{section}

\newcommand\xqed[1]{%
	\leavevmode\unskip\penalty9999 \hbox{}\nobreak\hfill
	\quad\hbox{#1}}
\newcommand\Endremark{\xqed{$\Diamond$}}

\makeatletter
\newcommand{\opnorm}{\@ifstar\@opnorms\@opnorm}
\newcommand{\@opnorms}[1]{%
	\left|\mkern-1.5mu\left|\mkern-1.5mu\left|
	#1
	\right|\mkern-1.5mu\right|\mkern-1.5mu\right|
}
\newcommand{\@opnorm}[2][]{%
	\mathopen{#1|\mkern-1.5mu#1|\mkern-1.5mu#1|}
	#2
	\mathclose{#1|\mkern-1.5mu#1|\mkern-1.5mu#1|}
}
\makeatother

\begin{document}

\let\origmaketitle\maketitle
\def\maketitle{
	\begingroup
	\def\uppercasenonmath##1{} 
	\let\MakeUppercase\relax 
	\origmaketitle
	\endgroup
}

\title{\bfseries \Large{ {Exponential Decay of Matrix $\mathrm{\Phi}$-Entropies on Markov Semigroups with Applications to Dynamical Evolutions of Quantum Ensembles}}}

\author{ \large \textsc{Hao-Chung Cheng$^{1,2}$, Min-Hsiu Hsieh$^2$, and Marco  Tomamichel$^3$}}
\address{\small  	
$^{1}$Graduate Institute Communication Engineering, National Taiwan University, Taiwan (R.O.C.)\\
	$^{2}$Centre for Quantum Computation and Intelligent Systems,\\
	Faculty of Engineering and Information Technology, University of Technology Sydney, Australia\\
	$^3$School  of  Physics,  The University  of  Sydney,  Sydney,  Australia}
\email{\href{mailto:F99942118@ntu.edu.tw}{F99942118@ntu.edu.tw}}
\email{\href{mailto:Min-Hsiu.Hsieh@uts.edu.au}{Min-Hsiu.Hsieh@uts.edu.au}}
\email{\href{mailto:marco.tomamichel@sydney.edu.au}{marco.tomamichel@sydney.edu.au}}


\begin{abstract}
In the study of Markovian processes, one of the principal achievements is the equivalence between the $\Phi$-Sobolev inequalities and an exponential decrease of the $\Phi$-entropies.
In this work, we develop a framework of Markov semigroups on matrix-valued functions and generalize the above equivalence to the exponential decay of matrix $\Phi$-entropies. This result also specializes to spectral gap inequalities and modified logarithmic Sobolev inequalities in the random matrix setting. To establish the main result, we define a non-commutative generalization of the carr\'e du champ operator, and prove a de Bruijn's identity for matrix-valued functions. 

The proposed Markov semigroups acting on matrix-valued functions have immediate applications in the characterization of the dynamical evolution of quantum ensembles. We consider two special cases of quantum unital channels, namely, the depolarizing channel and the phase-damping channel. In the former, since there exists a unique equilibrium state, we show that the matrix $\Phi$-entropy of the resulting quantum ensemble  decays exponentially as time goes on. Consequently, we obtain a stronger notion of monotonicity of the Holevo quantity---the Holevo quantity of the quantum ensemble decays exponentially in time and the convergence rate is determined by the modified log-Sobolev inequalities.  However, in the latter, the matrix $\Phi$-entropy of the quantum ensemble that undergoes the phase-damping Markovian evolution generally will not decay exponentially. This is because there are multiple equilibrium states for such a channel.  

Finally, we also consider examples of statistical mixing of Markov semigroups on matrix-valued functions. We can explicitly calculate the convergence rate of a Markovian jump process defined on Boolean hypercubes, and provide upper bounds of the mixing time on these types of examples.
\end{abstract}

\maketitle

\section{Introduction} \label{sec:introduction}

The core problem when studying dynamical systems is to understand how they evolve as time progresses. For example, we want to understand the equilibrium of a stochastic process. 
The \emph{Markov semigroup theory} mathematically describes the time evolution of dynamical systems. With a Markov semigroup operator $\{\mathsf{P}_t\}_{t\geq 0}$ acting on real-valued functions defined on some Polish space $\Omega$, for example, one can ask: is there an invariant measure $\mu$  such that
\[
\int_{x\in\Omega} \mathsf{P}_t f(x) \, \mu(\mathrm{d}x) = \int_{x\in\Omega}  f(x) \, \mu(\mathrm{d}x)
\]
holds for all such functions $f$?
(In the following we use the shorthand $\int f \,\mathrm{d}\mu$ for $\int_{x\in\Omega} f(x) \,\mu(\mathrm{d}x)$.) 
If there does exist such an invariant measure $\mu$ on $\Omega$, then how fast the system evolution $\mathsf{P}_t f $ converges to the constant equilibrium $\int f \,\mathrm{d}\mu$ when $t$ goes to infinity? To address these problems, functional inequalities like \emph{spectral gap inequalities} (also called \emph{Poincar\'e inequalities}) and \emph{logarithmic Sobolev inequalities} (log-Sobolev) play crucial roles
\cite{DS96, Bak94, Bak06, BGL13, Gui03, Sal97}.
More explicitly, the spectral gap inequality with a constant $C>0$:
\begin{align} \label{eq:Var2}
\Var(f) \triangleq \int f^2 \, \mathrm{d}\mu - \left( \int f \, \mathrm{d}\mu \right)^2 \leq C\mathcal{E}(f,f)
\end{align}
where $\Var(f)$ denotes the variance of the real-valued function $f$ and $\mathcal{E}(f,f)$ is the ``energy" of $f$  (see Section \ref{sec:semigroup} for formal definitions), is equivalent to the so-called $\mathbb{L}^2$ ergodicity of the semigroup $\{\mathsf{P}_t\}_{t\geq 0}$:
\begin{align} \label{eq:Var}
\Var(\mathsf{P}_t f) \leq \mathrm{e}^{-2t/C} \Var(f).
\end{align}
On the other hand, the log-Sobolev inequality, which is well-known from the seminal work of Gross \cite{Gro75}: 
\begin{align} \label{eq:Ent}
\Ent(f^2) \triangleq \int f^2\log \left( \frac{f^2}{\int f^2  \, \mathrm{d}\mu} \right) \, \mathrm{d}\mu  \leq C\mathcal{E}(f,f) 
\end{align}
is equivalent to an exponential decrease of entropies: 
\begin{align} \label{eq:Ent2}
\Ent(\mathsf{P}_t f) \leq \mathrm{e}^{-t/C} \Ent(f).
\end{align}

Chafa\"i \cite{Cha03} generalized the previous results and introduced the classical $\Phi$-entropy functionals to establish the equivalence between exponential decays of the $\Phi$-entropies to the $\Phi$-Sobolev inequalities, which interpolates between spectral gap and log-Sobolev inequalities \cite{LO00}. 
Consequently, the optimal constant in those functional inequalities directly determines the convergence rate of the Markov semigroups.

Recently, Chen and Tropp \cite{CT14} introduced a \emph{matrix $\Phi$-entropy functional} for the matrix-valued function $\bm{f}:\Omega \to \mathbb{C}^{d\times d}$, extending its classical counterpart to include matrix objects, and proved a \emph{subadditive} property. This extension has received great attention, and leads to powerful matrix concentration inequalities \cite{Tro15, PMT14}.  Furthermore, two of the present authors \cite{CH1} derived a series of matrix Poincar\'e inequalities and matrix $\Phi$-Sobolev inequalities for the matrix-valued functions. This result partially generalized Chafa\"i's work \cite{Cha03, Cha06}.

Equipped with the tools of matrix $\Phi$-entropies \cite{CT14} and the functional inequalities \cite{CH1}, we are at the position to explore a more general form of dynamical systems; namely those systems consisting of matrix components and their evolution governed by the Markov semigroup: 
\begin{align*} 
\mathsf{P}_t \bm{f}(x) = \int_{y\in\Omega} {\mathsf{T}}_t(x,\mathrm{d}y) \circ \bm{f}(y),
\end{align*}
where ${\mathsf{T}}_t(x,\mathrm{d}y):\mathbb{C}^{d\times d}\to \mathbb{C}^{d\times d}$ is a completely positive (CP) map and $\int_{y\in\Omega} {\mathsf{T}}_t(x,\mathrm{d}y)$ is unital. 
We are able to establish the equivalence conditions for the exponential decay of matrix $\Phi$-entropy functionals.

\medskip
The contributions of this paper are the following: 
\begin{enumerate}
\item We propose a Markov semigroup acting on matrix-valued functions and define a non-commutative version of the carr\'e du champ operator $\bm{\Gamma}$ in Section \ref{sec:semigroup}. 
We obtain the time derivatives of matrix $\Phi$-entropy functionals, a generalization of the de Bruijn's identity for matrix-valued functions in Proposition \ref{prop:de}. 
 The equivalence condition of the exponential decay of matrix $\Phi$-entropy functionals is established in Theorem \ref{theo:decay}. When $\Phi$ is a square function, our result generalizes Eqs.~\eqref{eq:Var2} and \eqref{eq:Var} to the equivalence condition of matrix spectral gap inequalities (Corollary \ref{coro:spectral}). On the other hand, when $\Phi(u)=u\log u$, we obtain the equivalence between exponential entropy decays and the modified log-Sobolev inequalities (Corollary \ref{coro:Ent}). This is slightly different from Eqs.~\eqref{eq:Ent} and \eqref{eq:Ent2}.


\item We show that the introduced Markov semigroup has a connection with quantum information theory and can be used to characterize the dynamical evolution of quantum ensembles that do not depend on the history.
More precisely, when the outputs of the matrix-valued function are restricted to a set of quantum states $\bm{f}(x) = \bm{\rho}_x$ (i.e.~positive semi-definite matrices with unit trace), the measure $\mu$ together with the function $\bm{f}$ yields a \emph{quantum ensemble} $\mathcal{S} \triangleq \left\{ \mu(x), \bm{\rho}_x \right\}_{x\in\Omega}$.
Its time evolution undergoing the semigroup $\{\mathsf{P}_t\}_{t\geq 0}$ can be described by $\mathcal{S}_t \triangleq \left\{ \mu(x), \int \mathsf{T}_t(x,\mathrm{d} y) \circ \bm{\rho}_y \right\}_{x\in\Omega}$.
Moreover, the matrix $\Phi$-entropy functional coincides the \emph{Holevo quantity} $\chi(\{\mu(x), \bm{\rho}_x\}_{x\in\Omega}) = \Ent(\bm{f})$.
Our main theorem hence shows that the Holevo quantity of the ensemble $\mathcal{S}$ exponentially decays through the dynamical process:
$
\chi\left( \mathcal{S}_t \right)
\leq \mathrm{e}^{t/C} \chi \left( \mathcal{S} \right),
$
where the convergence rate is determined by the modified log-Sobolev inequality\footnote{Here we assume there exists a unique \emph{invariant measure} (see Section \ref{sec:semigroup} for precise definitions) for the Markov semigroups, which ensures the existence of the average state. We discuss the conditions of uniqueness in Section \ref{sec:exam} and \ref{sec:exam2}.}.
This result directly strengthens the celebrated monotonicity of the Holevo quantity \cite{Pet03}.

\item We study an example of matrix-valued functions defined on a Boolean hypercube $\{0,1\}^n$  with transition rates $p$ from state $0$ to $1$ and $(1-p)$ from $1$ to $0$ (so-called \emph{Markovian jump process}). In this example, we can explicitly calculate the convergence rate of the Markovian jump process (Theorem \ref{theo:Var_Ber} and \ref{theo:Ent_Ber}) by exploiting the matrix Efron-Stein inequality \cite{CH1}.

\item We introduce a random walk of a quantum ensemble, where each vertex of the graph corresponds to a  quantum state and the transition rates are determined by the weights of the edges. 
The time evolution of the ensemble can be described by a statistical mixture of the density operators. By using the Holevo quantity as the entropic measure, our main theorem shows that the states in the ensemble will converge  to its equilibrium---the average state of the ensemble. Moreover, we can upper bound the mixing time of the ensemble.
\end{enumerate}

\subsection{Related Work} \label{ssec:related}
When considering the case of discrete time and finite domain (i.e.~$\Omega$ is a finite set), our setting reduces to the \emph{discrete quantum Markov} introduced by Gudder \cite{Gud08}, and the family $\{\mathsf{T}_t(x,y)\}$ is called the \emph{transition operation matrices} (TOM).
We note that this discrete model has been applied in quantum random walks by Attal \textit{et al.} \cite{APS+12, APS12, HJ14}, and model-checking in quantum protocols by Feng \textit{et al.} \cite{FYY13, YYF+13}.

If the state space is a singleton (i.e.~$\Omega = \{x\}$) and the trace-preserving property is imposed (i.e.~$\mathsf{T}_t$ is a quantum channel),  our model reduces to the conventional quantum Markov processes (also called the \emph{quantum dynamical semigroups}) \cite{Lin76, TKR+10, KRW12, SW13, SRW14} and 
the Markov semigroups defined on non-commutative $\mathbb{L}_p$ spaces \cite{OZ99}.
This line of research was initiated by Lindblad who studied the time evolution of a quantum state: $\mathsf{T}_t(\bm{\rho})$ and was recently extended by Kastoryano \textit{et al.} \cite{KRW12} and Szehr \cite{SW13, SRW14} who analyzed its long-term behavior.
Subsequently, Olkiewicz and Zegarlinski generalized Gross' log-Sobolev inequalities \cite{Gro75} to the non-commutative $\mathbb{L}_p$ space \cite{OZ99}. The connections between the quantum Markov processes, hypercontracitivity, and the noncommutative log-Sobolev inequalities are hence established \cite{BZ00, Car04, TPK14, Kin14, CKMT15, MFW15, SK15}.
The exponential decay properties in non-commutative $\mathbb{L}_p$ spaces are also studied \cite{TKR+10, KT13, Car14,CM15}. 

The major differences between our work and the non-commutative setting are the following: (1) the semigroup in the latter is applied on a single quantum state, i.e.~$\bm{\rho} \mapsto \mathsf{T}_t(\bm{\rho})$,
while the novelty of this work is to propose a Markov semigroup that acts on matrix-valued functions. i.e.~$\bm{f} \mapsto \mathsf{P}_t \bm{f}$. (2) Olkiewicz and Zegarlinski used a $\mathbb{L}_p$ relative entropy: 
\[
\Ent(\bm{\rho}) = \frac1d\Tr[\bm{\rho}\log \bm{\rho}] - \frac1d\Tr[\bm{\rho}]\log\left(\frac1d\Tr[\bm{\rho}]\right)
\]
to measure the state.
However, every state in the ensemble is endowed with a probability $\mu(x)$ in this work.  Thus, we can use the matrix $\Phi$-entropy functionals \cite{CT14}:
\[
H_\Phi(\bm{\rho}_X) = \sum_{x\in\Omega} \mu(x) \Tr\left[ \bm{\rho}_x \log \bm{\rho}_x \right] - \Tr\left[ \left( \sum_{x}\mu(x)\bm{\rho}_x \right) \log\left( \sum_{x}\mu(x)\bm{\rho}_x\right)   \right]
\]
as the measure of the ensemble through the dynamical process. 
In other words, we investigate the time evolution and the long-term behavior of a quantum ensemble instead of a quantum state.
The key tools to develop the whole theory require operator algebras (e.g.~Fr\'echet derivatives and operator convex functions), the subadditivity of the matrix $\Phi$-entropy functionals, operator Jensen inequalities \cite{HP03, FZ07}, and the matrix $\Phi$-Sobolev inequalities \cite{CH1}.  

\medskip
The paper is organized as follows. The notation and basic properties of matrix algebras are presented in Section \ref{sec:preliminaries}. The Markov semigroups acting on matrix-valued functions are introduced in Section \ref{sec:semigroup}. We establish the main results of exponential decays in Section \ref{sec:main}. In Section \ref{sec:exam} we study the quantum ensemble going through a quantum unital channel and demonstrate the exponential decays of the Holevo quantity. We discuss another example of a statistical mixture of the semigroup in Section \ref{sec:exam2}. We prove an upper bound to the mixing time of a quantum random graph.
Finally, we conclude this paper in Section \ref{sec:diss}.

\section{Preliminaries and Notation} \label{sec:preliminaries}

\subsection{Notation and Definitions} \label{ssec:notation}

We denote by $\mathbb{M}^\text{sa}$ the space of all self-adjoint operators on some (separable) Hilbert space.
If we restrict to the case of $d\times d$ Hermitian matrices, we refer to the notation $\mathbb{M}_d^\text{sa}$.
Denote by $\Tr$ the standard trace function.
The Schatten $p$-norm is defined by $\| \bm{M} \|_p \triangleq ( \Tr |\bm{M}|^p )^{1/p}$ for $1\leq p<\infty$, and $\|\bm{M}\|_\infty$ corresponds to the operator norm.

For $\bm{A},\bm{B}\in\mathbb{M}^\textnormal{sa}$, $\bm{A}\succeq \bm{B}$ means that $\bm{A}-\bm{B}$ is positive semi-definite. Similarly, $\bm{A} \succ \bm{B}$ means $\bm{A} - \bm{B}$ is positive-definite.
Denote by $\mathbb{M}^+$ (resp.~$\mathbb{M}_d^+$) the positive semi-definite operators (resp.~$d\times d$ positive semi-definite matrices).
Considering any matrix-valued function $\bm{f},\bm{g}:\Omega \to \mathbb{M}^\text{sa}$ defined on some Polish space $\Omega$\footnote{A Polish space is a separable and complete metric space, e.g.~a discrete space, $\mathbb{R}$, or the set of Hermitian matrices.}, we shorthand $\bm{f}-\bm{g} \succeq \bm{0}$ for $\bm{f}(x) - \bm{g}(x) \succeq \bm{0}$, $\forall x\in\Omega$.
Throughout the paper, italic boldface letters (e.g.~$\bm{X}$ or $\bm{f}$) are used to denote matrices or matrix-valued functions.

A linear map $\mathsf{T}:\mathbb{M}^\text{sa} \to \mathbb{M}^\text{sa}$ is \emph{positive} if $\mathsf{T}(\bm{A}) \succeq \bm{0}$ for all $\bm{A} \succeq \bm{0}$.
A linear map $\mathsf{T}:\mathbb{M}^\text{sa} \to \mathbb{M}^\text{sa}$ is \emph{completely positive} (CP)
if for any $\mathbb{M}_d^\text{sa}$, the map $\mathsf{T}\otimes \mathds{1}$ is positive on $\mathbb{M}^\text{sa} \otimes \mathbb{M}^\text{sa}$.
It is well-known that any CP map $\mathsf{T}:\mathbb{M}^\text{sa} \to \mathbb{M}^\text{sa}$ enables a \emph{Kraus decomposition}
\[
\mathsf{T}(\bm{A}) = \sum_{i} \bm{K}_i \bm{A} \bm{K}_i^\dagger.
\]
The CP map $\mathsf{T}$ is \emph{trace-preserving} (TP) if and only if $\sum_i \bm{K}_i^\dagger \bm{K}_i = \bm{I}$ (the identity matrix in $\mathbb{M}^\text{sa}$), and is \emph{unital} if and only if $\sum_i \bm{K}_i \bm{K}_i^\dagger = \bm{I}$ (see e.g.~\cite{MW09}). A CPTP map is often called a \emph{quantum channel} or \emph{quantum operation} in quantum information theory \cite{NC09}.
We denote by $|i-1\rangle \langle i-1|$ the zero matrix with the exception that its $i$-th diagonal entry is $1$. The set $\{|0\rangle, |1\rangle, \ldots, |d-1\rangle\}$ is the computational basis of the Hilbert space $\mathbb{C}^d$.

\begin{defn}[Matrix $\Phi$-Entropy Functional {\cite{CT14}}] \label{defn:entropy}
	Let $\mathrm{\Phi}:[0,\infty)\rightarrow \mathbb{R}$ be a convex function. Given any probability space $(\Omega, \Sigma, \mathbb{P})$, consider
	a positive semi-definite random matrix $\bm{Z}$ that is $\mathbb{P}$-measurable. Its expectation
	\[
	\mathbb{E}[\bm{Z}] \triangleq \int_\Omega \bm{Z}\, \mathrm{d} \mathbb{P} =   \int_{x\in\Omega} \bm{Z}(x)\, \mathbb{P} (\mathrm{d}x)
	\]
is a bounded matrix in $\mathbb{M}^+$.
	Assume $\bm{Z}$ satisfies the integration conditions: $\Tr\left[\mathbb{E}| \bm{Z}|\right]<\infty$ and $\Tr\left[\mathbb{E}| \mathrm{\Phi}(\bm{Z})|\right]<\infty$.
	The matrix $\mathrm{\Phi}$-entropy functional $H_\mathrm{\Phi}$ is defined as\footnote{ Chen and Tropp \cite[Definition 2.4]{CT14} defined the matrix $\Phi$-entropy functional for the random matrix $\bm{Z}$ taking values in $\mathbb{M}_d^+$ with the normalized trace function:
		\[
		H_\mathrm{\Phi}(\bm{Z})\triangleq \tr\left[\mathbb{E}\mathrm{\Phi}(\bm{Z})-\mathrm{\Phi}(\mathbb{E}\bm{Z})\right],
		\]
	where $\tr[\cdot] \triangleq \frac1d \Tr[\cdot]$. 
	In this paper we adopt the standard trace function; however, the results remain valid for $\tr$ as well.
	}
	\[ 
	H_\mathrm{\Phi}(\bm{Z})\triangleq \Tr\left[\mathbb{E}\mathrm{\Phi}(\bm{Z})-\mathrm{\Phi}(\mathbb{E}\bm{Z})\right].
	\]
	Define $\mathcal{F}\subseteq \Sigma$ as a sub-sigma-algebra of $\Sigma$ such that the expectation 
	$\mathbb{E}[\bm{Z}|\mathcal{F}]$ satisfying 
	$\int_E \mathbb{E}[\bm{Z}|\mathcal{F}] \,\mathrm{d}\mathbb{P} = \int_E \bm{Z} \, \mathrm{d}\mathbb{P}$ for each measurable set $E\in\mathcal{F}$.
	The conditional matrix $\Phi$-entropy functional is then
	\[
	H_\mathrm{\Phi}(\bm{Z}|\mathcal{F})\triangleq \Tr\left[\mathbb{E} \left[ \mathrm{\Phi}(\bm{Z})|\mathcal{F} \right] - \mathrm{\Phi}(\mathbb{E}\left[ \bm{Z} | \mathcal{F} \right] )\right].
	\]
	In particular, we denote by $\Ent(\bm{Z}) \triangleq H_\Phi(\bm{Z})$  when $\Phi(u)\equiv u\log u$ and call it the \emph{entropy functional}.
\end{defn}

\begin{theo}[Subadditivity of Matrix $\mathrm{\Phi}$-Entropy Functionals {\cite[Theorem 2.5]{CT14}}] \label{theo:sub}
	Let $X\triangleq (X_1, \ldots, X_n)$ be a vector of independent random variables taking values in a Polish space.
	Consider a positive semi-definite random matrix $\bm{Z}$ that can be expressed as a measurable function of the random vector $X$.
	Assume the integrability conditions  $\Tr\left[\mathbb{E}| \bm{Z}|\right]<\infty$ and $\Tr\left[\mathbb{E}| \mathrm{\Phi}(\bm{Z})|\right]<\infty$.
	If $\Phi(u)= u \log u$ or $\Phi(u) = u^p$ for $1\leq p\leq 2$, then
	\begin{eqnarray}\label{eq:entropy}
	H_\mathrm{\Phi}(\bm{Z})\leq \sum_{i=1}^n \mathbb{E} \Big[ H_\mathrm{\Phi}(\bm{Z}| X_{-i}) \Big],
	\end{eqnarray}
	where 
the random vector ${X}_{-i}\triangleq ({X}_1,\ldots,{X}_{i-1},{X}_{i+1},\ldots,{X}_n)$ is obtained by deleting the $i$-th entry of $X$.
\end{theo}

\subsection{Matrix Algebra} \label{ssec:matrix}

Let $\mathcal{U},\mathcal{W}$ be Banach spaces.
The \emph{Fr\'{e}chet derivative} of a function $\mathcal{L}:\mathcal{U} \rightarrow \mathcal{W}$ at a point $\bm{X}\in\mathcal{U}$, if it exists\footnote{We assume the functions considered in the paper are Fr\'{e}chet differentiable. The readers is referred to~\cite{Pel85,Bic12} for conditions for when a function is Fr\'{e}chet differentiable. }, is a unique linear mapping $\mathsf{D}\mathcal{L}[\bm{X}]:\mathcal{U}\rightarrow\mathcal{W}$ such that
\[
\lim_{\|\bm{E}\|_\mathcal{U} \rightarrow 0} \frac{\|\mathcal{L}(\bm{X}+\bm{E}) - \mathcal{L}(\bm{X}) - \mathsf{D}\mathcal{L}[\bm{X}](\bm{E})\|_{\mathcal{W}}}{\|\bm{E}\|_\mathcal{U}} = 0,
\]
where $\|\cdot\|_{\mathcal{U}(\mathcal{W})}$ is a norm in $\mathcal{U}$ (resp.~$\mathcal{W}$). The notation $\mathsf{D}\mathcal{L}[\bm{X}](\bm{E})$ then is interpreted as ``the Fr\'{e}chet derivative of $\mathcal{L}$ at $\bm{X}$ in the direction $\bm{E}$''.

The Fr\'echet derivative enjoys several properties of usual derivatives.
\begin{prop}[Properties of Fr\'{e}chet Derivatives {\cite[Section 5.3]{AH09}}] \label{prop:properties}
	Let $\mathcal{U},\mathcal{V}$ and $\mathcal{W}$ be real Banach spaces.
	\begin{itemize}
		\item[1.] (Sum Rule) If $\mathcal{L}_1:\mathcal{U}\rightarrow\mathcal{W}$ and $\mathcal{L}_2:\mathcal{U}\rightarrow\mathcal{W}$ are Fr\'{e}chet differentiable at $\bm{A}\in\mathcal{U}$, then so is $\mathcal{L} = \alpha \mathcal{L}_1 + \beta \mathcal{L}_2$ and $\mathsf{D}\mathcal{L}[\bm{A}](\bm{E}) = \alpha \cdot \mathsf{D}\mathcal{L}_1[\bm{A}](\bm{E}) + \beta \cdot \mathsf{D}\mathcal{L}_2[\bm{A}](\bm{E})$.
		\item[2.] (Product Rule) If $\mathcal{L}_1:\mathcal{U}\rightarrow\mathcal{W}$ and $\mathcal{L}_2:\mathcal{U}\rightarrow\mathcal{W}$ are Fr\'{e}chet differentiable at $\bm{A}\in\mathcal{U}$ and assume the multiplication is well-defined in $\mathcal{W}$, then so is $\mathcal{L}=\mathcal{L}_1 \cdot \mathcal{L}_2$ and $\mathsf{D}\mathcal{L}[\bm{A}](\bm{E}) = \mathsf{D}\mathcal{L}_1[\bm{A}](\bm{E}) \cdot \mathcal{L}_2(\bm{A}) + \mathcal{L}_1(\bm{A}) \cdot \mathsf{D}\mathcal{L}_2[\bm{A}](\bm{E})$.
		\item[3.] (Chain Rule)  Let $\mathcal{L}_1:\mathcal{U}\rightarrow\mathcal{V}$ and $\mathcal{L}_2:\mathcal{V}\rightarrow\mathcal{W}$ be Fr\'{e}chet differentiable at $\bm{A}\in\mathcal{U}$ and $\mathcal{L}_1(\bm{A})$ respectively, and let $\mathcal{L} = \mathcal{L}_2 \circ \mathcal{L}_1$ (i.e.~$\mathcal{L}(\bm{A}) = \mathcal{L}_2\left( \mathcal{L}_1 (\bm{A}) \right)$. Then $\mathcal{L}$ is Fr\'{e}chet differentiable at $\bm{A}$ and $\mathsf{D}\mathcal{L}[\bm{A}](\bm{E}) = \mathsf{D}\mathcal{L}_2 [\mathcal{L}_1(\bm{A})] \left( \mathsf{D}\mathcal{L}_1[\bm{A}](\bm{E}) \right)$.
	\end{itemize}
\end{prop}


For each self-adjoint and bounded operator $\bm{A}\in\mathbb{M}^\textnormal{sa}$ with the spectrum $\sigma(\bm{A})$ and the spectral measure $\bm{E}$, its \emph{spectral decomposition} can be written as
$\bm{A} = \int_{\lambda \in \sigma(\bm{A})} \lambda \, \mathrm{d} \bm{E}(\lambda)$.
As a result, each scalar function is extended to a \emph{standard matrix function} as follows:
\[
f(\bm{A}) \triangleq  \int_{\lambda \in \sigma(\bm{X})} f(\lambda) \, \mathrm{d} \bm{E}(\lambda).
\]
A real-valued function $f$ is called \emph{operator convex} if for each $\bm{A},\bm{B}\in \mathbb{M}^\textnormal{sa}$ and $0\leq t \leq 1$,
\[
f(t\bm{A}) + f((1-t)\bm{B}) \preceq f( t\bm{A} + (1-t)\bm{B}).
\]

\begin{prop}
	[Operator Jensen's Inequality for Matrix-Valued Measures \cite{HP03}, {\cite[Theorem 4.2]{FZ07}}] \label{prop:Jensen}
	Let $(\mathrm{\Omega},\mathrm{\Sigma})$ be a measurable space and suppose that $I\subseteq \mathbb{R}$ is an open interval. 
	Assume for every $x \in \Omega$, $\bm{K}(x)$ is a (finite or infinite dimensional) square matrix and satisfies 
	\[ \int_{x\in\Omega} \bm{K}(\mathrm{d}x) \bm{K}(\mathrm{d}x)^\dagger  = \bm{I} \] (identity matrix in $\mathbb{M}^\text{sa}$).
	If $\bm{f}:\mathrm{\Omega}\rightarrow \mathbb{M}^\text{sa}$ is a measurable function for which $\sigma(\bm{f}(x)) \subset I$, for every $x\in \Omega$, then
	\[
	\phi\left( \int_{x\in\Omega} \bm{K}(\mathrm{d}x) \bm{f}(x) \bm{K}(\mathrm{d}x)^\dagger  \right) \preceq  \int_{x\in\Omega} \bm{K}(\mathrm{d}x) \phi\left(\bm{f}(x)\right) \bm{K}(\mathrm{d}x)^\dagger \, \mu(\mathrm{d}x)
	\]
	for every operator convex function $\phi:I\rightarrow \mathbb{R}$.
	Moreover, 
	\[
	\Tr \left[ \phi\left( \int_{x\in\Omega}  \bm{K}(\mathrm{d}x) \bm{f}(x) \bm{K}(\mathrm{d}x)^\dagger \, \mu(\mathrm{d}x) \right) \right] \leq \Tr \left[ \int_{x\in\Omega} \bm{K}(\mathrm{d}x) \phi\left(\bm{f}(x)\right) \bm{K}(\mathrm{d}x)^\dagger \, \mu(\mathrm{d}x) \right]
	\]
	for every convex function $\phi:I\rightarrow \mathbb{R}$.
\end{prop}

\begin{prop}[{\cite[Theorem 3.23]{HP14}}] \label{prop:trace_Petz}
	Let $ \bm{A}, \bm{X}\in\bbM^{sa}$ and $t\in\bbR$. Assume $f:I\to \bbR$ is a continuously differentiable function defined on an interval $I$ and assume that the eigenvalues of $ \bm{A}+t\bm{X} \subset I$. Then
	\[
	\left.\frac{\mathrm{d} }{ \mathrm{d} t} \Tr f(\bm{A}+t \bm{X})\right|_{t=t_0} = \Tr [  \bm{X} f' ( \bm{A} + t_0 \bm{X}) ].
	\]
\end{prop}

\section{Markov Semigroups on Matrix-Valued Functions} \label{sec:semigroup}

We will introduce the theory of Markov semigroups in this section. We particularly focus on the Markov semigroup acting on matrix-valued functions. The reader may find general references of Markov semigroups on real-valued functions in \cite{Bak94, Bak06, BGL13, Gui03, Sal97}.

Throughout this paper, we consider a probability space $(\mathrm{\Omega},\mathrm{\Sigma},\mu)$ with $\Omega$ being a discrete space or a compact connected smooth manifold (e.g.~$\Omega\equiv\mathbb{R}$).
We consider the Banach space $\mathcal{B}$ of continuous, bounded and Bochner integrable (\cite{Die77, Mik78}) matrix-valued functions $\bm{f}:\mathrm{\Omega}\rightarrow \mathbb{M}^\text{sa}$ equipped with the uniform norm (e.g.~$\opnorm{\bm{f}} \triangleq \sup_{x\in\Omega }\|  \bm{f}(x) \|_\infty$).
We denote the expectation with respect to the measure $\mu$ for any measurable function $\bm{f}:\Omega \to \mathbb{M}^\text{sa}$ by
\[
\mathbb{E}_{\mu}[\bm{f}] \triangleq \int_\Omega \bm{f} \, \mathrm{d} \mu = \int_{x\in\Omega} \bm{f}(x) \, \mu(\mathrm{d}x).
\]
where the integral is the Bochner integral.
We also instate the integration condition $\Tr\left[ \mathbb{E}_{\mu} | \bm{f} | \right]\leq \infty$.

Define a \emph{completely positive kernel} $ \mathsf{T}_t (x, \mathrm{d} y): \mathbb{M}^\text{sa} \to \mathbb{M}^\text{sa}$ to be a family of CP maps on $\mathbb{M}^\text{sa}$ depending on the parameter $t\in\mathbb{R}_+\triangleq[0,\infty)$ such that 
\begin{align} \label{eq:label}
\int_{y\in\Omega} \mathsf{T}_t (x, \mathrm{d} y) \;\text{is a unital map},
\end{align}
and satisfies the \emph{Chapman-Kolmogorov identity}:
\begin{align} \label{eq:kernel}
\int_{y\in\mathrm{\Omega}} \mathsf{T}_s(x,\mathrm{d}y) \circ  \mathsf{T}_t(y,\mathrm{d}z) = \mathsf{T}_{s+t} (x,\mathrm{d}z)\quad \forall s,g\in\mathbb{R}_+.
\end{align}
In particular, we can impose the trace-preserving property: $\int_{y\in\Omega} \mathsf{T}_t (x, \mathrm{d} y)$ is a unital quantum channel.

The central object investigated in this work is a family of operators $\{\mathsf{P}_t\}_{t\geq 0}$ acting on matrix-valued functions. These operators are called \emph{Markov semigroups} if they satisfy:
\begin{defn}
	[Markov Semigroups on Matrix-Valued Functions] \label{defn:Markov}
	A family of linear operators $\{\mathsf{P}_t\}_{t\geq 0}$ on the Banach space $\mathcal{B}$ is a \emph{Markov semigroup} if and only if it satisfies the following conditions
	\begin{itemize}
		\item[(a)] $\mathsf{P}_0 = \mathds{1}$, the identity map on $\mathcal{B}$ (\emph{intitial condition}).
		\item[(b)] The map $t\to \mathsf{P}_t \bm{f}$ is a continuous map from $\mathbb{R}_+$ to $\mathcal{B}$ (\emph{continuity property}).
		\item[(c)] The semigroup properties: $\mathsf{P}_t \circ \mathsf{P}_s = \mathsf{P}_{s+t}$, for any $s,t\in\mathbb{R}_+$ (\emph{semigroup property}).
		\item[(d)] $\mathsf{P}_t \bm{I} = \bm{I}$ for any $t\in\mathbb{R}_+$, where $\bm{I}$ is the constant identity matrix in $\mathbb{M}^\text{sa}$ (\emph{mass conservation}).
		\item[(e)] If $\bm{f}$ is non-negative (i.e.~$\bm{f}(x)\succeq \bm{0}$ for all $x\in\Omega$), then $\mathsf{P}_t \bm{f}$ is non-negative for any $t\in\mathbb{R}_+$ (\emph{positivity preserving}).
	\end{itemize}
\end{defn}
Given the CP kernel $\{\mathsf{T}_t(x,\mathrm{d}y)\}_{t\geq 0}$, the Markov semigroup acing on the matrix-valued function $\bm{f}$ is defined by
\begin{align} \label{eq:Pt}
\mathsf{P}_t \bm{f}(x) \triangleq \int_{y\in\Omega} {\mathsf{T}}_t(x,\mathrm{d}y) \bm{f}(y),
\end{align}
where $\mathsf{T}_t(x,\mathrm{d}y) \bm{f}(y)$ refers to the output of the linear map $\mathsf{T}_t(x,\mathrm{d}y)$ acting on $\bm{f}(y)$.
If $\int_{y\in\Omega} \mathsf{T}_t (x, \mathrm{d} y)$ is a CPTP map, then it exhibits the \emph{contraction} property of the semigroup (with respect to the norm $\opnorm{\,\cdot\,}$ on the Banach space $\mathcal{B}$):
\begin{align} \label{eq:contraction}
\begin{split}
\opnorm{ \mathsf{P}_t \bm{f}  } &= \sup_{x\in\Omega} \left\| \mathsf{P}_t \bm{f}(x) \right\|_\infty  \\
&= \sup_{x\in\Omega} \left\| \int_{y\in\Omega} \mathsf{T}_t(x,\mathrm{d} y) \bm{f}(y)  \right\|_\infty  \\
&\leq \sup_{y\in\Omega} \left\| \bm{f}(y) \right\|_\infty \\
&= \opnorm{ \bm{f} },
\end{split}
\end{align}
where we use fact that the CPTP map is contractive \cite{PWP+06}.


From the operator Jensen's inequality (Proposition \ref{prop:Jensen}) we have the following two inequalities:
\begin{align} \label{eq:Pt_convex}
\mathsf{P}_t \left( \phi \circ \bm{f} \right) = 
\mathsf{P}_t \phi\left( \bm{f} \right) \succeq  \phi\left( \mathsf{P}_t\bm{f} \right)
\end{align}
for any operator convex function $\phi$, and
\begin{align} \label{eq:Pt_convex2}
\Tr \left[ \mathsf{P}_t \phi\left( \bm{f} \right) \right] \geq \Tr \left[  \phi\left( \mathsf{P}_t\bm{f} \right) \right]
\end{align}
for any convex function $\phi$.

Since the map $t\to \mathsf{P}_t \bm{f}$ is continuous (Definition \ref{defn:Markov}); hence, the derivative of the operator $\mathsf{P}_t$ with respect to $t$, i.e. the convergence rate of $\mathsf{P}_t$, is a main focus of the analysis of Markov semigroups. 
More precisely, 
we define the \emph{infinitesimal generator} for any Markov semigroup $\{\mathsf{P}_t\}_{t\geq0}$ by
\begin{align} \label{eq:L}
\mathsf{L}(\bm{f}) \triangleq \lim_{t\rightarrow 0^+} \frac1t( \mathsf{P}_t \bm{f} - \bm{f} ).
\end{align}
For convenience, we denote by $\mathcal{D}(\mathsf{L})$ the \emph{Dirichlet domain} of $\mathsf{L}$, which is the set of matrix-valued functions in $\mathcal{B}$ such that the limit in Eq.~\eqref{eq:L} exists. We provide an equivalent condition of $\mathcal{D}(\mathsf{L})$ in Appendix \ref{sec:HY}.

Combined with the linearity of the operators $\{\mathsf{P}_t\}_{t\geq 0 }$ and the semigroup property, we  deduce that the generator $\mathsf{L}$ is the derivative of $\mathsf{P}_t$ at any time $t> 0$.
That is, for $t,s>0$,
\[
\frac1s \left( \mathsf{P}_{t+s} - \mathsf{P}_t \right)
= \mathsf{P}_t \left( \frac1s \left( \mathsf{P}_s - \mathds{1} \right) \right)
= \left( \frac1s \left( \mathsf{P}_s - \mathds{1} \right) \right) \mathsf{P}_t.
\]
Letting $s\to 0$ shows that
\begin{align} \label{eq:partial_Pt}
\frac{\partial}{\partial t} \mathsf{P}_t = \mathsf{L} \mathsf{P}_t = \mathsf{P}_t \mathsf{L}.
\end{align}
The above equation combined with Eq.~$\eqref{eq:Pt_convex}$ implies the following proposition.
\begin{prop} \label{prop:L}
	Let $\{\mathsf{P}_t \}_{t\geq 0}$ be a Markov semigroup with the infinitesimal generator $\mathsf{L}$.
	For any operator convex function $\phi:\mathbb{R}\to\mathbb{R}$ and $\bm{f} \in \mathcal{D}(\mathsf{L})$, we have
	\begin{align} \label{eq:L_fre}
	\mathsf{L}\left( \phi(\bm{f}) \right) \succeq \mathsf{D}\phi[\bm{f}]\left( \mathsf{L} \bm{f} \right).
	\end{align}
\end{prop}

\begin{proof}
For any $s>0$, Eq.~\eqref{eq:Pt_convex} implies
\begin{align} \label{eq:L_fre1}
\frac1s \big( \mathsf{P}_{s} \phi(  \bm{f} ) - \phi( \mathsf{P}_0 \bm{f} ) \big) \succeq \frac1s \big( \phi( \mathsf{P}_{s} \bm{f} ) - \phi( \mathsf{P}_0 \bm{f} ) \big).
\end{align}
By letting $s\to 0^+$ and using the chain rule of the Fr\'echet derivative (Proposition \ref{prop:properties}), the right-hand side yields
	\begin{align} \label{eq:L_fre2}
	\begin{split}
	\lim_{s\to 0^+} \frac1s \big( \phi( \mathsf{P}_{0+s} \bm{f} ) - \phi( \mathsf{P}_0 \bm{f} ) \big)
	&= \left. \mathsf{D} \phi \left[ \mathsf{P}_t \bm{f} \right] \left( 
	\frac{\partial}{\partial t} \mathsf{P}_t \bm{f} \right) \right|_{t=0} \\
	&= \left. \mathsf{D} \phi \left[ \mathsf{P}_t \bm{f} \right] \left( 
	 \mathsf{L} \mathsf{P}_t \bm{f} \right) \right|_{t=0} \\
	&= \mathsf{D} \phi \left[\bm{f} \right] \left( 
	 \mathsf{L} \bm{f} \right),	
	\end{split}
	\end{align}
	where the second equality follows from Eq.~\eqref{eq:partial_Pt}. In the last line we apply the property $\mathsf{P}_0 \bm{f} = \bm{f}$ (item (a) in Definition \ref{defn:Markov}).
	
On the other hand, the left-hand side of Eq.~\eqref{eq:L_fre1} can be rephrased as
	\begin{align} \label{eq:L_fre3}
	\begin{split}
	\lim_{s\to 0^+} \frac1s \big( \mathsf{P}_{s} \phi(  \bm{f} ) - \phi( \mathsf{P}_0 \bm{f} ) \big)
	&= 	\lim_{s\to 0^+} \frac1s \big( \mathsf{P}_{s} \phi(  \bm{f} ) - \phi(  \bm{f} ) \big) \\
	&= \mathsf{L} \left( \phi(\bm{f}) \right).
	\end{split}
	\end{align}
	Hence, combining Eqs.~\eqref{eq:L_fre3}, \eqref{eq:L_fre1} and \eqref{eq:L_fre2} arrives at the desired inequality
	\begin{align*}
	\mathsf{L}\left( \phi(\bm{f}) \right) \succeq \mathsf{D}\phi[\bm{f}]\left( \mathsf{L} \bm{f} \right).
	\end{align*}
\end{proof}


In the classical setup (i.e.~Markov semigroups acting on real-valued functions), the \emph{carr\'{e} du champ}  operator (see e.g.~\cite{BGL13}) is defined by :
\begin{align} \label{eq:Gamma0}
\Gamma(f,g) = \frac12 \big( \mathsf{L}(fg) - f \mathsf{L}(g) - g \mathsf{L}(f) \big).
\end{align}
Here we introduce a non-commutative version of the {carr\'{e} du champ} operator $\mathbf{\Gamma}:\mathcal{D}(\mathsf{L})\times \mathcal{D}(\mathsf{L}) \rightarrow \mathcal{D}(\mathsf{L})$ of the generator $\mathsf{L}$ by
\begin{align} \label{eq:Gamma1}
\mathbf{\Gamma}(\bm{f},\bm{f}) \triangleq \frac12 \big( \mathsf{L}(\bm{f}^2) - \bm{f}\mathsf{L}(\bm{f}) - \mathsf{L}(\bm{f})\bm{f}\big),
\end{align}
and its symmetric and bilinear extension 
\begin{align} \label{eq:Gamma2}
\begin{split}
\mathbf{\Gamma}(\bm{f},\bm{g}) = \mathbf{\Gamma}(\bm{g},\bm{f}) &\triangleq \frac12 \big( \mathbf{\Gamma}(\bm{f}+\bm{g},\bm{f}+\bm{g}) - \mathbf{\Gamma}(\bm{f},\bm{f}) - \mathbf{\Gamma}(\bm{g},\bm{g}) \big) \\
&= \frac14 \left( \mathsf{L}( \bm{f} \bm{g} ) + \mathsf{L}( \bm{g} \bm{f} ) - \bm{f}\mathsf{L}(\bm{g}) - \bm{g}\mathsf{L}(\bm{f}) - \mathsf{L}(\bm{f})\bm{g} - \mathsf{L}(\bm{g})\bm{f} \right).
\end{split}
\end{align}
We note that when $\bm{f}$ commutes\footnote{Here we means that $[\bm{f}(x),\bm{g}(x)]=\bm{0}$ for all $x\in\mathrm{\Omega}$.} with $\bm{g}$, the carr\'{e} du champ operator reduces to the conventional expression (cf.~Eq.~\eqref{eq:Gamma0}):
\begin{align*} 
\mathbf{\Gamma}(\bm{f},\bm{g}) \equiv \frac12 \left( \mathsf{L}(\bm{f}\bm{g}) - \bm{f}\mathsf{L}(\bm{g}) - \bm{g}\mathsf{L}(\bm{f})\right).
\end{align*}
Recall that the square function $\phi(u) = u^2$ is operator convex.
The formula of the Fr\'echet derivative: $\mathsf{D}\phi[\bm{A}](\bm{B}) = \bm{A}\bm{B} + \bm{B}\bm{A}$ together with Proposition \ref{prop:L} yields
\[
\mathsf{L}(\bm{f}^2) \succeq \bm{f}\mathsf{L}(\bm{f}) + \mathsf{L}(\bm{f})\bm{f}.
\]
Hence the carr\'{e} du champ operator is positive semi-definite:
$\mathbf{\Gamma}(\bm{f},\bm{f}) \succeq \bm{0}$.
We can also observe that $\bm{\Gamma}(\bm{f},\bm{f}) = \bm{0}$ implies that $\bm{f}$ is essentially constant, i.e.~$\mathsf{P}_t \bm{f} = \bm{f}$ for $t\geq 0$.
Moreover, the non-negativity and the bilinearity of the carr\'{e} du champ operator directly yield a trace Cauchy-Schwartz inequality:
\begin{prop}
	[Trace Cauchy-Schwartz Inequality for Carr\'e du Champ Operators] \label{prop:Cauchy}
	For all $\bm{f},\bm{g} \in \mathcal{D}(\mathsf{L})$,
	\[
	\big( \Tr \left[ \bm{\Gamma}(\bm{f},\bm{g}) \right]  \big)^2 \leq \Tr \left[ \bm{\Gamma}(\bm{f},\bm{f}) \right] \cdot \Tr \left[ \bm{\Gamma}(\bm{g},\bm{g}) \right].
	\]
\end{prop}
\begin{proof}
	From Eq.~\eqref{eq:Gamma2}, for all $s\in \mathbb{R}$  if follows that
	\begin{align*}
	\bm{\Gamma}( s\bm{f} + \bm{g}, s\bm{f} + \bm{g}) 
	&=  s^2 \cdot \bm{\Gamma}( \bm{f}, \bm{f}) + 2s \cdot \bm{\Gamma}( \bm{f}, \bm{g} ) + \bm{\Gamma}(\bm{g},\bm{g}) \succeq \bm{0}
	\end{align*}
	After taking trace, the non-negativity of the above equation ensures the discriminant being non-positive:
	\[
	 \big( 2\cdot \Tr \left[ \bm{\Gamma}(\bm{f},\bm{g}) \right]  \big)^2 - 4 \cdot \Tr \left[ \bm{\Gamma}(\bm{f},\bm{f}) \right] \cdot \Tr \left[ \bm{\Gamma}(\bm{g},\bm{g}) \right] \leq 0
	\]
	as desired.
\end{proof}
Note that a bilinear map $ \langle \cdot, \cdot \rangle : V \times V \to \mathbb{R}$ on some vector space $V$ is a \emph{scalar inner product} if it satisfies conjugate symmetry and non-negativity ($\langle x, x\rangle \geq0$ and $x=0$ when $\langle x, x\rangle =0$).
As a result, the non-commutative carr\'e du champ operator that exhibits the properties of symmetry, linearity, and non-negativity ($\bm{\Gamma}(\bm{f},\bm{f}) \succeq \bm{0}$ and $\mathsf{L}\bm{f} = \bm{0}$ when $\bm{\Gamma}(\bm{f},\bm{f}) = \bm{0}$)
can be viewed as a matrix-valued inner product (with respect to the generator $\mathsf{L}$) on the space $\mathcal{D}(\mathsf{L})$.

Given the semigroup $\{\mathsf{P}_t\}_{t\geq 0}$,
the measure $\mu$ is \emph{invariant} for the function $\bm{f}\in \mathcal{D}(\mathsf{L})$ if 
\begin{align} \label{eq:invariance}
\int \mathsf{P}_t \bm{f}(x) \, \mu(\mathrm{d}x) = \int \bm{f}(x) \,\mu(\mathrm{d} x) , \quad t\in\mathbb{R}_+.
\end{align}
We can observe from Eqs.~\eqref{eq:L} and \eqref{eq:invariance} that any invariant measure $\mu$ satisfies
\begin{align} \label{eq:L0}
\int \mathsf{L}(\bm{f}) \, \mathrm{d} \mu = \bm{0}.
\end{align}
We call the measure $\mu$ \emph{symmetric} if and only if
\begin{align} \label{eq:symmetric}
\int \bm{f} \mathsf{L}(\bm{g}) +  \mathsf{L}(\bm{g}) \bm{f} \, \mathrm{d} \mu
= \int \bm{g} \mathsf{L}(\bm{f}) +  \mathsf{L}(\bm{f}) \bm{g} \, \mathrm{d} \mu,
\end{align}
which implies an integration by parts formula:
\begin{align} \label{eq:by_part}
-\frac12 \int \bm{f} \mathsf{L}(\bm{g}) +  \mathsf{L}(\bm{g}) \bm{f} \, \mathrm{d} \mu
= \int \mathbf{\Gamma}(\bm{f},\bm{g}) \, \mathrm{d} \mu
= \int \mathbf{\Gamma}(\bm{g},\bm{f}) \, \mathrm{d} \mu.
\end{align}

The notion of carr\'{e} du champ operator and the invariant measure $\mu$ (i.e.~Eq.~\eqref{eq:L0}) immediately lead to a symmetric bilinear \emph{Dirichlet form}:
\begin{align} \label{eq:Dirichlet}
\begin{split}
\bm{\mathcal{E}}(\bm{f},\bm{g}) \triangleq \int \mathbf{\Gamma}(\bm{f},\bm{g}) \, \mathrm{d} \mu
&= -\frac14 \int \bm{f}\mathsf{L}(\bm{g}) + \bm{g}\mathsf{L}(\bm{f}) + \mathsf{L}(\bm{f})\bm{g} + \mathsf{L}(\bm{g})\bm{f} \, \mathrm{d}\mu.
\end{split}
\end{align}
The non-negativity of the carr\'e du champ operator also yields that
$\bm{\mathcal{E}}(\bm{f},\bm{f}) = \int \mathbf{\Gamma}(\bm{f},\bm{f}) \, \mathrm{d} \mu \succeq \bm{0}$, which stands for a kind of second-moment quantity or the energy of the function $\bm{f}$.
For convenience, we shorthand the notation $\mathbf{\Gamma}(\bm{f}) \equiv \mathbf{\Gamma}(\bm{f},\bm{f})$ and $\bm{\mathcal{E}}(\bm{f})\equiv \bm{\mathcal{E}}(\bm{f},\bm{f})$.

In the following, we often refer to the \emph{Markov Triple} $(\Omega,\mathbf{\Gamma},\mu)$ with state space $\Omega$, carr\'e du champ operator $\mathbf{\Gamma}$ acting on the Dirichlet domain $\mathcal{D}(\mathsf{L})$ of matrix-valued functions, and invariant measure $\mu$.
Additionally, we will apply the Fubini's theorem to freely interchange the order of trace and the expectation with respect to $\mu$.

\section{Main Results: Exponential Decays of Matrix $\Phi$-Entropy Functionals} \label{sec:main}

In this section, our goal is to show that the matrix $\Phi$-entropy functional exponentially decays along the Markov semigroup and its relation with the spectral gap inequalities and logarithmic Sobolev inequalities.
With the invariant measure $\mu$ of the semigroup $\left\{\mathsf{P}_t\right\}_{t\geq0}$ and the Jensen inequality \eqref{eq:Pt_convex2}, we observe that
\begin{align*}
H_\Phi \left(\mathsf{P}_t \bm{f} \right) &= \Tr \Big[ \mathbb{E}_{\mu}\big[\Phi \left(\mathsf{P}_t \bm{f} \right) \big]  - \Phi\big( \mathbb{E}_{\mu} \left[ \mathsf{P}_t \bm{f} \right] \big) \Big] \\
&=  \Tr \Big[ \mathbb{E}_{\mu}\big[\Phi \left(\mathsf{P}_t \bm{f} \right) \big]  - \Phi\big( \mathbb{E}_{\mu}  \bm{f}  \big) \Big] \\
&\leq \Tr \Big[ \mathbb{E}_{\mu}\big[\mathsf{P}_t\Phi \left( \bm{f} \right) \big]  - \Phi\big( \mathbb{E}_{\mu}  \bm{f}  \big) \Big] \\
&= H_\Phi \left( \bm{f} \right),
\end{align*}
where in the second and the last lines we use the property of the invariant measure $\mu$, Eq.~\eqref{eq:invariance}.
Thus, the matrix $\Phi$-entropy functional is non-increasing along the flow of the semigroup and behaves like the classical $\Phi$-entropy functionals (see e.g.~\cite{Cha03}).
Moreover, we are able to obtain the time differentiation of the matrix $\Phi$-entropy functional, which can be viewed as the \emph{Boltzmann H-Theorem} for matrix-valued functions.
\begin{prop}
	[de Bruijn's Property for Markov Semigroups] \label{prop:de}
	Fix a probability space $(\Omega,\Sigma,\mu)$.
	Let $\left\{ \mathsf{P}_t \right\}_{t\geq 0} $ be a Markov semigroup with infinitesimal generator $\mathsf{L}$ and carr\'e du champ operator $\mathbf{\Gamma}$.
	Assume that $\mu$ is an invariant probability measure for the semigroup.
	Then, for any suitable matrix-valued function $\bm{f}:\Omega \rightarrow \mathbb{M}^\text{sa}$ in the Dirichlet domain $\mathcal{D}(\mathsf{L})$ with $\mu$ being its invariant measure,
	\begin{align} \label{eq:deBru1}
	\frac{\partial}{\partial t} H_\Phi \left( \mathsf{P}_t \bm{f} \right) 
	= \Tr \mathbb{E}_{\mu} \big[ \Phi'\left(\mathsf{P}_t \bm{f}\right)  \mathsf{L}\mathsf{P}_t \bm{f}  \big] \leq 0,\quad \forall t\in\mathbb{R}_+.
	\end{align}
	When $\mu$ is symmetric, one has the following formulation
	\begin{align} \label{eq:deBru2}
	\frac{\partial}{\partial t} H_\Phi \left( \mathsf{P}_t \bm{f} \right) = 
	- \Tr\mathbb{E}_{\mu} \big[ \mathbf{\Gamma} \left(\Phi' \left( \mathsf{P}_t \bm{f} \right), \mathsf{P}_t \bm{f} \right) \big],\quad \forall t\in\mathbb{R}_+.
	\end{align}
\end{prop}

\begin{proof}
	The proof directly follows from the definition of the matrix $\Phi$-entropy functional and the properties of the Markov semigroup.
	Namely,
	\begin{align*}
	\frac{\partial}{\partial t} H_\Phi \left( \mathsf{P}_t \bm{f} \right)
	&= \frac{\partial}{\partial t} \Tr \Big[ \mathbb{E}_{\mu} \big[ \Phi \left( \mathsf{P}_t \bm{f} \right) \big] - \Phi \big( \mathbb{E}_{\mu}\left[ \mathsf{P}_t \bm{f} \right] \big) \Big] \\ 
	&= \frac{\partial}{\partial t} \Tr \Big[ \mathbb{E}_{\mu} \big[ \Phi \left( \mathsf{P}_t \bm{f} \right) \big] - \Phi \big( \mathbb{E}_{\mu}\left[ \bm{f} \right] \big) \Big] \\
	&= \frac{\partial}{\partial t} \Tr\mathbb{E}_{\mu} \big[ \Phi \left( \mathsf{P}_t \bm{f} \right) \big] \\
	&= \Tr \mathbb{E}_{\mu} \big[ \mathsf{D}\Phi\left[\mathsf{P}_t \bm{f}  \right] \big( \mathsf{L}\mathsf{P}_t \bm{f} \big) \big]\\
	&=  \Tr \mathbb{E}_{\mu} \big[ \Phi'\left(\mathsf{P}_t \bm{f}\right)  \mathsf{L}\mathsf{P}_t \bm{f}  \big],
	\end{align*}
where the second equality is due to the invariance of $\mu$, Eq.~\eqref{eq:invariance}. The fourth equation is due to the chain rule of Fr\'echet derivative (see Proposition \ref{prop:properties}) and Eq.~\eqref{eq:partial_Pt}.	We obtain the last identity by Proposition \ref{prop:trace_Petz}.
	
Proposition \ref{prop:L} yields $\mathsf{D}\Phi\left[\mathsf{P}_t \bm{f}  \right] \big( \mathsf{L}\mathsf{P}_t \bm{f} \big) \preceq \mathsf{L}\Phi(\mathsf{P}_t \bm{f})$.
	By the invariance of $\mu$, we deduce the non-positivity of Eq.~\eqref{eq:deBru1}:
	\begin{align*}
	\mathbb{E}_{\mu} \big[  \mathsf{D}\Phi\left[\mathsf{P}_t \bm{f}  \right] \big( \mathsf{L}\mathsf{P}_t \bm{f} \big)
	\big] \preceq \mathbb{E}_{\mu} \big[  \mathsf{L}\Phi(\mathsf{P}_t \bm{f}) \big]
	= \bm{0}.
	\end{align*}
	
	The symmetric case \eqref{eq:deBru2} stands by further applying the integration by parts formula, Eq.~\eqref{eq:by_part}, i.e.~
	\begin{align*}
	\Tr \Big[ \mathbb{E}_{\mu} \big[\bm{\Gamma}(\bm{g},\bm{h})   \big]  \Big]
	&= -\frac12 \Tr \Big[ \mathbb{E}_{\mu} \big[ \bm{g} \cdot \mathsf{L}(\bm{h}) + \mathsf{L}(\bm{h}) \cdot \bm{g}   \big]\Big] \\
	&= - \Tr \Big[ \mathbb{E}_{\mu} \big[ \bm{g} 
	\cdot \mathsf{L}(\bm{h}) \big] \Big],
	\end{align*}
where we apply the cyclic property of the trace function. Hence, Eq.~\eqref{eq:deBru2} follows by taking $\bm{g} \equiv \Phi'(\mathsf{P}_t \bm{f})$ and $\bm{h} \equiv \mathsf{P}_t \bm{f}$.
\end{proof}

In the following, we first give the definitions of the spectral gap inequalities and logarithmic Sobolev inequalities related to Markov semigroups. The main result---the relation between the exponential decays in matrix $\Phi$-entropies and the functional inequalities, is presented in Theorem \ref{theo:decay}.

\begin{defn}
	[Spectral Gap Inequality for Matrix-Valued Functions] \label{defn:spectral}
	A Markov Triple $(\Omega,\mathbf{\Gamma},\mu)$ is said to satisfy a spectral gap inequality with a constant $C>0$, if for all matrix-valued functions $\bm{f}:\Omega \rightarrow \mathbb{M}_d^\text{sa}$ in the Dirichlet domain $\mathcal{D}(\mathsf{L})$ with $\mu$ being its invariant measure,
	\[
	\textnormal{Var}(\bm{f}) \leq C\mathcal{E}(\bm{f}),
	\]
	where 
	\[
	\textnormal{Var}(\bm{f})\triangleq \Tr\mathbb{E}_{\mu}\left[  \big( \bm{f}-  \mathbb{E}_{\mu}[\bm{f}] \big)^2  \right]
	\]
denotes the variance of the function $\bm{f}$ with respect to the measure $\mu$ and $\mathcal{E}(\bm{f})\triangleq \Tr\left[ \bm{\mathcal{E}}(\bm{f}) \right]$.
The infimum of the constants among all the spectral gap inequalities is called the \emph{spectral gap constant}.
\end{defn}

\begin{defn}
[Logarithmic Sobolev Inequality for Matrix-Valued Functions] \label{defn:log}
A Markov Triple $(\Omega,\mathbf{\Gamma},\mu)$ is said to satisfy a logarithmic Sobolev inequality LS$(C,B)$ with constants $C>0$, $B\geq 0$, if for all matrix-valued functions $\bm{f}:\Omega \rightarrow \mathbb{M}_d^\text{sa}$ in the Dirichlet domain $\mathcal{D}(\mathsf{L})$  with $\mu$ being its invariant measure,
	\[
	\Ent\left( \bm{f}^2 \right) \leq B \, \Tr\mathbb{E}_{\mu} \big[ \bm{f}^2 \big] + C\mathcal{E}(\bm{f}).
	\]
	The logarithmic Sobolev inequality is called \emph{tight} and is denoted by LS$(C)$ when $B=0$.
	When $B>0$, the logarithmic Sobolev inequality LS$(C,B)$ is called \emph{defective}.
	
	We also define the \emph{modified logarithmic Sobolev inequality} (MLSI) if there exists a constant $C$ such that
	\[
	\Ent\left( \bm{f} \right) \leq -C \Tr\mathbb{E}_{\mu} \big[\left( \bm{I} + \log\bm{f}\right)  \mathsf{L} \bm{f}  \big].
	\]
\end{defn}

\begin{theo}
	[Exponential Decay of Matrix $\Phi$-Entropy Functionals of Markov Semigroups] \label{theo:decay}
	Given a Markov triple $(\Omega, \bm{\Gamma}, \mu)$, the following two statements are equivalent: there exists a $\Phi$-Sobolev constant $C \in (0,\infty]$ such that 
\begin{align} \label{eq:exp1}
H_\Phi(\bm{f}) \leq -C \Tr \mathbb{E}_{\mu} \big[ \Phi'\left( \bm{f}\right)  \mathsf{L} \bm{f}  \big], 
\end{align}
and 
\begin{align} \label{eq:exp2}
H_\Phi( \mathsf{P}_t\bm{f}) \leq \mathrm{e}^{-t/C} H_\Phi(\bm{f}), \quad \forall t\geq 0
\end{align}
for all $\bm{f}\in\mathcal{D}(\mathsf{L})$ with $\mu$ being its invariant measure.
\end{theo}
\begin{proof}
	The theorem is a consequence of the de Bruijn's property, Proposition \ref{prop:de}.
	More precisely, Eq.~\eqref{eq:deBru1} and the inequality \eqref{eq:exp1} imply
	\begin{align*}
	\frac{\partial}{\partial t} H_\Phi \left( \mathsf{P}_t \bm{f} \right) 
	&= \Tr \mathbb{E}_{\mu} \big[ \Phi'\left(\mathsf{P}_t \bm{f}\right)  \mathsf{L}\mathsf{P}_t \bm{f}  \big].
	\end{align*}
	Recall that the function $\mathsf{P}_t\bm{f}$ is invariant under the measure $\mu$ and therefore satisfies Eq.~\eqref{eq:exp1}.
	Hence, the above inequality can be rewritten as
	\begin{align*}
		\frac{\partial}{\partial t} H_\Phi \left( \mathsf{P}_t \bm{f} \right) 
		&= \Tr \mathbb{E}_{\mu} \big[ \Phi'\left(\mathsf{P}_t \bm{f}\right)  \mathsf{L}\mathsf{P}_t \bm{f}  \big]\\
		&\leq -\frac1C H_\Phi(\mathsf{P}_t \bm{f}),
	\end{align*}
	from which we obtain 
	\[
	H_\Phi( \mathsf{P}_t\bm{f}) \leq \mathrm{e}^{-t/C} H_\Phi( \mathsf{P}_0\bm{f}) =\mathrm{e}^{-t/C} H_\Phi(\bm{f}).
	\]
	Conversely, differentiating inequality \eqref{eq:exp2} at $t=0$ gives the desired inequality \eqref{eq:exp1}.
\end{proof}

From Theorem \ref{theo:decay}, we immediately establish the equivalence between the spectral gap inequality and the exponential decay of variance functions.

\begin{coro}[Exponential Decay of Variance and Spectral Gap Inequalities] \label{coro:spectral}
	A Markov Triple $(\Omega,\mathbf{\Gamma},\mu)$ satisfies the spectral gap inequality with constant $C$ if and only if, for matrix-valued functions $\bm{f}:\Omega\rightarrow \mathbb{M}_d^\text{sa}$ in the Dirichlet domain $\mathcal{D}(\mathsf{L})$, one has
	\[
	\textnormal{Var}\big( \mathsf{P}_t \bm{f} \big) \leq \mathrm{e}^{-2t/C} \cdot \textnormal{Var}\left( \bm{f} \right).
	\]
\end{coro}
\begin{proof}
	Recall that $H_{u\mapsto u^2} (\bm{f}) \equiv \Var(\bm{f})$.
	Hence, by taking $\Phi(u)=u^2$ in Theorem \ref{theo:decay}, the corollary follows since the right-hand side of Eq.~\eqref{eq:exp1} can be rephrased as:
	\[
	\Tr \mathbb{E}_{\mu} \big[ \Phi'\left( \bm{f}\right)  \mathsf{L} \bm{f}  \big]
	= 2 \Tr \mathbb{E}_{\mu} \big[ \bm{f} \cdot  \mathsf{L} \bm{f}  \big]
	= - 2 \Tr\mathbb{E}_{\mu} \big[ \mathbf{\Gamma}\left(\bm{f}\right) \big]
	= -2 \mathcal{E}(\bm{f}).
	\]	
\end{proof}

Similarly, by taking $\Phi(u) = u\log u$, we have the equivalence between the modified log-Sobolev inequalities and  the exponential decay in entropy functionals.

\begin{coro}[Exponential Decay of Entropy and Modified Log-Sobolev Inequalities] \label{coro:Ent}
	A Markov Triple $(\Omega,\mathbf{\Gamma},\mu)$ satisfies the modified log-Sobolev inequality with constant $C$ if and only if, for matrix-valued functions $\bm{f}:\Omega\rightarrow \mathbb{M}_d^\text{sa}$ in the Dirichlet domain $\mathsf{D}(\mathsf{L})$, one has
	\[
	\textnormal{Ent}\big( \mathsf{P}_t \bm{f} \big) \leq \mathrm{e}^{-t/C} \cdot \textnormal{Ent}\left( \bm{f} \right).
	\]
\end{coro}

\section{Time Evolutions of a Quantum Ensemble} \label{sec:exam}
In this section, we discuss the applications of the Markov semigroups in quantum information theory and study a special case of the Markov semigroups---quantum unital channel.
From the analysis in Section \ref{sec:main}, we demonstrate the exponential decays of the matrix $\Phi$-entropy functionals and give a tight bound to the monotonicity of the Holevo quantity:
$\chi(\{\mu(x),\mathsf{T}_t (\bm{\rho}(x))\}_x) \leq \mathrm{e}^{-t/C}\chi(\{\mu(x),\bm{\rho}(x)\}_x)$, 
when $\mathsf{T}_t$ is a unital quantum dynamical semigroup \cite{Lin76}.

To connect the whole machinery to quantum information theory, it is convenient to introduce some basic notation. 
First of all, we will restrict to the following set of matrix-valued functions in this section: $\bm{f}(x) = \bm{\rho}_x$, where $\bm{\rho}_x$  is a density operator (i.e.~positive semi-definite matrix with unit trace) for all $x\in\Omega$. In other words, the function $\bm{f}$ is a classical-quantum encoder that maps the classical message $x$ to a quantum state $\bm{\rho}_x$. 
Therefore, a quantum ensemble is constructed by the classical-quantum encoder and the given measure $\mu$: $\mathcal{S}\triangleq \{ \mu(x), \bm{\rho}_x\}_{x\in\Omega}$.
The Holevo quantity of a quantum ensemble $\mathcal{S}$ is defined as
\begin{align*}
\chi(\mathcal{S}) 
&\triangleq -\Tr\big[ \overline{\bm{\rho}} \log \overline{\bm{\rho}} \, \big] 
+ \int_{x\in\Omega}  \Tr \big[ \bm{\rho}_x \log \bm{\rho}_x \big] \, \mu(\mathrm{d}x),
\end{align*}
where $\overline{\bm{\rho}} \triangleq \int \bm{\rho}_x \, \mu(\mathrm{d}x)$ denotes the average state. It is not hard to verify that the Holevo quantity is a special case of the matrix $\Phi$-entropy functionals \cite{CH1}: $
\chi(\mathcal{S})= \Ent(\bm{f}).$

If we impose the trace-preserving condition on the CP kernel $\int_{y\in\Omega} \mathsf{T}_t (x, \mathrm{d} y)$, it becomes a quantum unital channel. Consequently, the Markov semigroup $\{\mathsf{P}_t\}_{t\geq 0}$ acts on the matrix-valued function $\bm{f}$ can be interpreted as a time evolution of the quantum ensemble $\mathcal{S}$, and the results of the matrix-valued functions in Section \ref{sec:main} also work for quantum ensembles.

For each $t\in\mathbb{R}_+$, let the CP kernel be 
\begin{align*}
\mathsf{T}_t(x,y) = \begin{cases}
\mathsf{T}_t \; (\text{a quantum unital channel}), \; &\text{if } x = y \\
\mathsf{0} \; (\text{a zero map}), \; &\text{if } x\neq y.
\end{cases}
\end{align*}
The set of unital maps $\mathsf{T}_t: \mathbb{M}_d^\text{sa} \to \mathbb{M}_d^\text{sa}$ forms a \emph{quantum dynamical semigroup}, which satisfies the semigroup conditions:
\begin{itemize}
	\item[(a)] $\mathsf{T}_0 (\bm{X})= \bm{X}$ for all $\bm{X}\in\mathbb{M}_d^\text{sa}$.
	\item[(b)] The map $t\to \mathsf{T}_t (\bm{X})$ is a continuous map from $\mathbb{R}_+$ to $\mathbb{M}_d^\text{sa}$.
	\item[(c)] The semigroup properties: $\mathsf{T}_t \circ \mathsf{T}_s = \mathsf{T}_{s+t}$, for any $s,t\in\mathbb{R}_+$.
	\item[(d)] $\mathsf{T}_t (\bm{I}) = \bm{I}$ for any $t\in\mathbb{R}_+$, where $\bm{I}$ is the identity matrix in $\mathbb{M}_d^\text{sa}$ (\emph{mass conservation}).
	\item[(e)] $\mathsf{T}_t$ is a positive map for any $t\in\mathbb{R}_+$.
\end{itemize}
We note that the quantum dynamical semigroup has been studied in the contexts of quantum Markov processes (see Section \ref{ssec:related}).
It is shown \cite{Lin76, GKS76, AZ15} that any unital quantum dynamical semigroup is generated by a Liouvillian $\mathcal{L}: \mathbb{M}_d^\text{sa} \to \mathbb{M}_d^\text{sa}$ of the form
\begin{align*}
\mathcal{L}: \bm{X} \mapsto \mathsf{\Psi}(\bm{X}) - \kappa\bm{X} - \bm{X}\kappa^\dagger,
\end{align*}
where $\kappa \in \mathbb{C}^{d\times d}$ and $\mathsf{\Psi}$ is a CP map such that $\mathsf{\Psi}(\bm{I}) = \kappa + \kappa^\dagger$.
Therefore, each unital map can be expressed as $\mathsf{T}_t = \mathrm{e}^{t\mathcal{L}}$ for all $t\in\mathbb{R}_+$.
The Markov semigroup acting on the matrix-valued function $\bm{f}$ is hence defined by $\mathsf{P}_t \bm{f} (x) = \mathsf{T}_t ( \bm{f}(x))$ for all $x\in\Omega$.
The invariant measure exists if 
\begin{align*}
\int \bm{f}(x) \, \mu(\mathrm{d} x) &= 
\int \mathsf{P}_t \bm{f}(x) \, \mu(\mathrm{d}x) \\
&= \int \mathsf{T}_t ( \bm{f}(x)) \, \mu(\mathrm{d}x) \\
&= \mathsf{T}_t \left( \int \bm{f}(x) \, \mu(\mathrm{d}x) \right), \quad  \forall t\in\mathbb{R}_+.
\end{align*}
In other words, the expectation $\mathbb{E}_\mu[\bm{f}] = \int \bm{f}(x) \, \mu(\mathrm{d} x)$ is the fixed point of the unital semigroup $\{\mathsf{T}_t\}_{t\geq0}$.
From our main result---Theorem \ref{theo:decay}, we can establish the exponential decays of the matrix $\Phi$-entropy functionals through the quantum dynamical semigroup $\{\mathsf{T}_t\}_{t\geq 0}$:
\begin{align*}
H_\Phi(\bm{f}) \leq -C \Tr \mathbb{E}_\mu \big[ \Phi'\left( \bm{f}\right)  \mathsf{L} \bm{f}  \big]
\quad\text{if and only if} \quad
H_\Phi\left( \mathsf{T}_t(\bm{f})\right) \leq \mathrm{e}^{-t/C} H_\Phi(\bm{f}), \quad \forall t\geq 0,
\end{align*}
where the infinitesimal generator is given by $\mathsf{L}\bm{f}(x) = \mathcal{L}(\bm{f}(x))$, for all $x \in \Omega$.

In the following, we consider the cases of depolarizing and phase-damping channels, and demonstrate the exponential decay phenomenon when all the density operators converge to the same equilibrium.
However, as it will be shown in the case of the phase-damping channel, the $\Phi$-Sobolev constant is infinite when the density operators converge to different states.

\subsection{Depolarizing Channel}
	Denote by $\bm{\pi} \triangleq \bm{I}/d$ the maximally mixed state on the Hilbert space $\mathbb{C}^d$, and let $r>0$ be a constant.
	The quantum dynamical semigroup defined by the depolarizing channel is:
	\begin{align} \label{eq:pd}
	\mathsf{T}_t : \bm{f}(x) \mapsto \mathrm{e}^{-rt} \bm{f}(x) + \left( 1- \mathrm{e}^{-rt} \right) \Tr[\bm{f}(x)] \cdot \bm{\pi}, \quad \forall x\in\Omega,\, t\in\mathbb{R}_+.
	\end{align}
	It is not hard to verify that $\{\mathsf{T}_t\}_{t\geq 0}$ forms a Markov semigroup with a unique fixed point (also called the stationary state) $\Tr[\bm{f}]\bm{\pi}$.
	We assume $\mu$ is the invariant measure of $\bm{f}$: $\mathsf{T}_t\left(\mathbb{E}_\mu[\bm{f}]\right) = \mathbb{E}_\mu[\bm{f}] = \Tr[\bm{f}] \bm{\pi}$.
	The infinitesimal generator and the Dirichlet form can be calculated as
	\begin{align*}
	\mathsf{L} \bm{f} = \lim_{t\to 0} \frac1t \left( \mathsf{T}_t (\bm{f}) - \bm{f} \right)
	= r\left( \Tr[\bm{f}]\bm{\pi} - \bm{f} \right);
	\end{align*}
	\begin{align*}
	\bm{\mathcal{E}}(\bm{f}) = \frac12 \,\mathbb{E}_\mu \left[ \mathsf{L}\bm{f}^2 - \bm{f}\cdot\mathsf{L}\bm{f} - \mathsf{L}\bm{f}\cdot \bm{f} \right] 
	= \frac{r}2\left( \mathbb{E}_\mu[\bm{f}^2] + \Tr\mathbb{E}_\mu[\bm{f}^2]\bm{\pi} - 2\left(\Tr[\bm{f}]\bm{\pi} \right)^2 \right).
	\end{align*}
	Now let the matrix-valued function correspond to a set of density operators---$\bm{f}(x):=\bm{\rho}_x$, $\forall x\in\Omega$, with the average state $\overline{\bm{\rho}} = \bm{\pi}$.
	The constant $C_2$ in the spectral gap inequality is 
	\begin{align} \label{eq:Var_const}
	C_2 = \sup_{\bm{f}:\,\mathbb{E}_\mu[\bm{f}] = \bm{\pi}} \; \frac{\Var(\bm{f})}{\mathcal{E}(\bm{f})}
	= \sup_{\bm{\rho}_X:\,\overline{\bm{\rho}} = \bm{\pi}} \;
	\frac{2}{r} \cdot \frac{\Tr\mathbb{E}_\mu[\bm{\rho}_X^2] - \frac1d}{2\Tr\mathbb{E}_\mu[\bm{\rho}_X^2]  - \frac2d} = \frac1r,
	\end{align}
	where we denote by $X$ the random variable such that $\Pr(X=x)= \mu(x)$ for all $x\in\Omega$.
	Hence, the spectral gap constant is $C_2=\frac1r$, and we have the exponential decays of the variance from Corollary \ref{coro:spectral}:
	\begin{align} \label{eq:pd_var}
	\Var(\mathsf{T}_t(\bm{\rho}_X)) \leq \mathrm{e}^{-2rt} \cdot \Var(\bm{\rho}_X).
	\end{align}
	
	Similarly, the modified log-Sobolev constant $C_\chi$ can be calculated by
	\begin{align} \label{eq:Ent_const}
	\begin{split}
	C_\chi &= \sup_{\bm{f}:\,\mathbb{E}_\mu[\bm{f}] = \bm{\pi}} \; \frac{\Ent(\bm{f})}{-\Tr\mathbb{E}_\mu\left[ (\bm{I}+\log\bm{f})\mathsf{L}\bm{f}\right]} \\
	&= \sup_{\bm{\rho}_X:\,\overline{\bm{\rho}} = \bm{\pi}} \;
	\frac1{r}\cdot \frac{\Tr\mathbb{E}_\mu[\bm{\rho}_X\log \bm{\rho}_X] + \log d}{\Tr\mathbb{E}_\mu[\bm{\rho}_X\log \bm{\rho}_X] - \frac{\Tr\mathbb{E}_\mu[\log \bm{\rho}_X]}{d}}.
	\end{split}
	\end{align}
	In the following proposition , we show that $C_\chi = \frac{1}{2r}$ when $d=2$. Therefore, we are able to establish the exponential decay of the Holevo quantity.
	\begin{prop} \label{prop:depolarizing}
		Consider a Hilbert space $\mathbb{C}^2$. Denote the quantum dynamical semigroup of  the depolarizing channel by
		\begin{align*}
		\mathsf{T}_t : \bm{\rho} \mapsto \mathrm{e}^{-rt} \bm{\rho} + \left( 1- \mathrm{e}^{-rt} \right) \cdot \bm{\pi}, \quad t\in\mathbb{R}_+.
		\end{align*}
		For any quantum ensemble on $\mathbb{C}^2$ with the average state being $\bm{\pi}$, the modified log-Sobolev constant is $C_\chi = \frac{1}{2r}$.
		Moreover, we have
		\begin{align} \label{eq:pd_Ent}
		\chi\left(\{\mu(x), \mathsf{T}_t\left( \bm{\rho}_x\right)\}_{x\in\Omega} \right) 
		\leq
		\mathrm{e}^{-2rt} \cdot \chi(\{\mu(x), \bm{\rho}_x\}_{x\in\Omega}).
		\end{align}
	\end{prop} 
	The proof can be found in Appendix \ref{proof:Prop15}.
\medskip
	We simulate the depolarizing qubit channel with the initial states $\bm{\rho}_1 = |0\rangle \langle 0|$, $\bm{\rho}_2 = |1\rangle \langle 1|$ and the uniform distribution in Figure \ref{fig:depolarizing}.
	The blue dashed curve and red solid curve show that the upper bounds for the exponential decays in Eqs.~\eqref{eq:pd_Ent} and \eqref{eq:pd_var} are quit tight.
	
	We remark that if every state $\bm{\rho}_x$ in the ensemble goes through different depolarizing channel with rate $r_x$, i.e.
	\begin{align*}
	\mathsf{T}_t(x,x) = \mathsf{T}_t^x : \bm{f}(x) \mapsto \mathrm{e}^{-r_x t} \bm{f}(x) + \left( 1- \mathrm{e}^{-r_x t} \right) \Tr[\bm{f}(x)] \cdot \bm{\pi},
	\end{align*}
	the Sobolev constant $C_2$ and $C_\chi$ will be dominated by the channel with the minimal rate $\inf_{x\in\Omega} r_x := r_\text{inf}$.
	Namely, the spectral gap constant in Eq.~\eqref{eq:Var_const} becomes
	\begin{align*}
	C_2 = \sup_{\bm{\rho}_X:\, \overline{\bm{\rho}} = \bm{\pi}} \;
	 \frac{\Tr\mathbb{E}_\mu[ \bm{\rho}_X^2] - \frac1d}{\Tr\mathbb{E}_\mu[r_X\bm{\rho}_X^2]  - \frac{\mathbb{E}_\mu[r_X]}d}
	\leq \frac{1}{r_\text{inf}},
	\end{align*}
	and the modified log-Sobolev constant in Eq.~\eqref{eq:Ent_const} is
	\begin{align*}
	C_\chi 
	\leq \frac{1}{r_\text{inf}} \cdot \frac{\Tr\mathbb{E}_\mu[\bm{\rho}_X\log \bm{\rho}_X] - \log d}{\Tr\mathbb{E}_\mu[\bm{\rho}_X\log \bm{\rho}_X] - \frac{\Tr\mathbb{E}_\mu[\log \bm{\rho}_X]}{d}}.
	\end{align*}

\subsection{Phase-Damping Channel}
	Fix $d=2$, and denote the Pauli matrix by
	\begin{align*}
	\bm{\sigma}_Z = \begin{pmatrix}
	1 & 0 \\ 0 & -1
	\end{pmatrix}.
	\end{align*}
The quantum dynamical semigroup defined by the phase-damping channel is
	\begin{align*}
	\mathsf{T}_t : \bm{f}(x) \mapsto \frac{\left(1+\mathrm{e}^{-rt}\right)}2 \bm{f}(x) + \frac{\left(1-\mathrm{e}^{-rt}\right)}2 \bm{\sigma}_Z\bm{f}(x)\bm{\sigma}_Z, \quad \forall x\in\Omega,\, t\in\mathbb{R}_+
	\end{align*}
	with the  generator:
	\begin{align*} 
	\mathsf{L}\bm{f} = \lim_{t\to 0} \frac{1-\mathrm{e}^{-rt}}{2t} \left( \bm{\sigma}_Z \bm{f} \bm{\sigma}_Z - \bm{f} \right)
	= \frac{r}2 \left( \bm{\sigma}_Z \bm{f} \bm{\sigma}_Z - \bm{f} \right).
	\end{align*}
It is well-known that any diagonal matrix (with respect to the computation basis) is a fixed point of the phase-damping channel $\mathsf{T}_t$.
	Now if we assume every matrix $\mathsf{P}_t\bm{f}(x)$ converges to different matrices, i.e.~$\mathsf{T}_t(\bm{f}(x)) \neq \mathsf{T}_t(\bm{f}(y))$ for all $x\neq y$ and $t\in\mathbb{R}_+$,
	then the matrix $\Phi$-entropy functional $H_\Phi(\mathsf{P}_t \bm{f})$ is non-zero for all $t\in\mathbb{R}_+$.
	However, the infinitesimal generator approaches zero as $t$ goes to infinity, i.e.~
	\begin{align*}
	\lim_{t\to \infty} \mathsf{L} \mathsf{P}_t \bm{f} 
	= \lim_{t\to \infty}\frac{r}2 \left( \bm{\sigma}_Z \left(\mathsf{P}_t \bm{f}\right) \bm{\sigma}_Z - \mathsf{P}_t \bm{f}  \right) = \bm{0}, 
	\end{align*}
	which means that the $\Phi$-Sobolev constant $C$ in Theorem \ref{theo:decay} is infinity.
	In other words, the matrix $\Phi$-entropy $H_\Phi(\mathsf{P}_t \bm{f})$ does not decay exponentially in this phase-damping channel.

\begin{remark}
The reason that makes these two examples quite different is the uniqueness of fixed point of the quantum dynamic semigroup $\mathsf{T}_t$. Since the depolarizing channel has a unique equilibrium state, all the matrices eventually converges. Hence, the Sobolev constants are finite, which leads to the exponential decay phenomenon. On the other hand, the phase-damping channel has multiple fixed points. This ensures the matrix $\Phi$-entropy functionals never vanish.
\end{remark}

\begin{figure}[ht]
	\includegraphics[width=0.8\columnwidth]{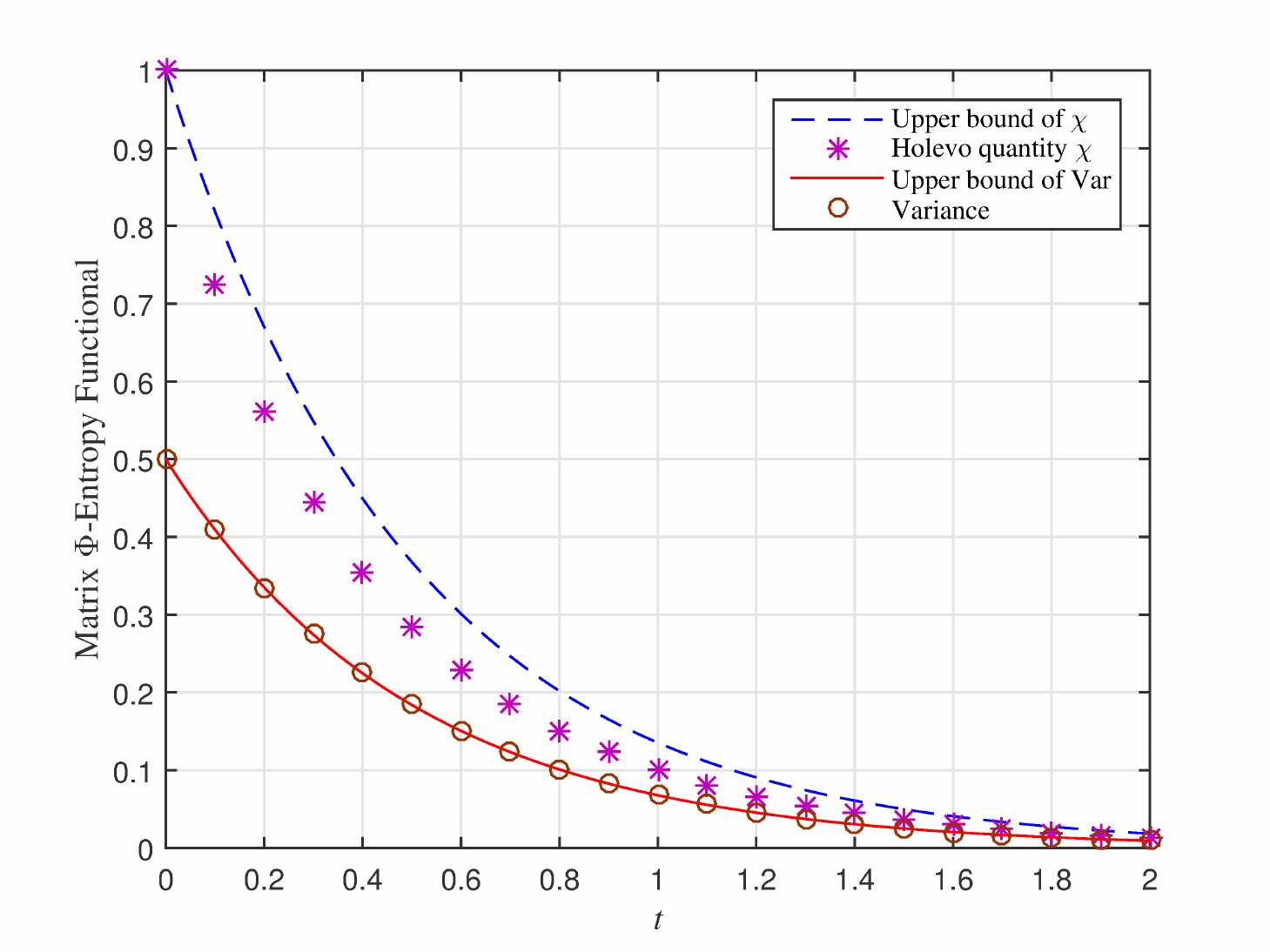}
	\caption{This figure illustrates the exponential decay phenomenon of the variance and the Holevo quantity through the depolarizing qubit channel $(d=2)$ (see Eq.~\eqref{eq:pd}) with rate $r=1$.
		Assume $\bm{\rho}_1 = |0\rangle\langle 0|$ and $\bm{\rho}_2 = |1\rangle\langle 1|$ with uniform distribution.
		The blue dashed curve and red solid curve represent the upper bounds of the Holevo quantity and the variance, respectively, i.e.~right-hand sides of Eqs.\eqref{eq:pd_Ent} and ~\eqref{eq:pd_var}.
		The actual variance and Holevo quantity through the time evolution of the depolarizing qubit channel
		are plotted by the `o' and `*' lines,
		, which demonstrates the tightness of the exponential upper bounds.
	}
	\label{fig:depolarizing}
\end{figure}

\section{The Statistical Mixture of the Markov Semigroup} \label{sec:exam2}
In this section, we study a statistical mixing of Markov semigroups.
The interested matrix-valued functions are defined on the Boolean hypercube $\{0,1\}^n$, which arise in the context of Fourier analysis \cite{Wol08, FS08}.  Moreover,  a matrix hypercontractivity inequality has been established on this particular set of matrix-valued functions \cite{BRW08}.  

Our first example is the {Markovian jump process} with transition rates $p$ from state $0$ to $1$ and $(1-p)$ from $1$ to $0$. We will calculate its convergence rate using the matrix Efron-Stein inequality \cite{CH1}.
Second, we consider the statistical mixing of a quantum random graph where each vertex corresponds to a quantum state, and further bound its mixing time.

\subsection{Markovian Jump Process on Symmetric Boolean Hypercube} \label{ssec:Jump}

We consider a special case of the Markov semigroup induced by a classical Markov kernel:
\begin{align} \label{eq:c_Markov}
\mathsf{P}_t \bm{f}(x) \triangleq \int_{y\in\Omega} {p}_t(x,\mathrm{d}y) \bm{f}(y),
\end{align}
where $p_t(x,\mathrm{d} y)$ is a family of transition probabilities\footnote{For every $t\geq0$ and $x\in\Omega$, $p_t(x,\cdot)$ is a probability measure on $\Omega$, and $x\mapsto p_t(x,E)$ is measurable for every measurable set $E\in\Sigma$} on $\Omega$, and satisfies the following Chapman-Kolmogorov identity
\[
\int_{y\in\mathrm{\Omega}} p_s(x,\mathrm{d}y) \,  p_t(y,\mathrm{d}z) = p_{s+t} (x,\mathrm{d}z).
\]
In other words, the time evolution of the matrix-valued function $\bm{f}$ is under a statistical mixture according to Eq.~\eqref{eq:c_Markov}. 
Let the state space be a hypercube, i.e.~$\Omega\equiv \{0,1\}^n$ with the measure denoted by 
\[
\mu_{n,p}(x) = p^{\sum_{i=1}^n x_i} (1-p)^{\sum_{i=1}^n (1-x_i)},  \quad \forall x\in\{0,1\}^n.
\]
We introduce the operator $\mathsf{\Delta}_i$ that acts on any matrix-valued function $\bm{f}:\{0,1\}^n \to\mathbb{M}_d^\text{sa}$ as follows:
\begin{align*}
\mathsf{\Delta}_i \bm{f} &= \begin{cases} (1-p) \mathsf{\nabla}_i \bm{f}, &\text{if } x_i = 1\\
-p \mathsf{\nabla}_i \bm{f}, &\text{if } x_i = 0
\end{cases} \\
&= \bm{f} - \int \bm{f}\, \mathrm{d} \mu_{1,p}(x_i),
\end{align*}
where 
\[
\mathsf{\nabla}_i \bm{f} \triangleq \bm{f}(x_1,\ldots,x_{i-1},1,x_{i+1},\ldots,x_n) - \bm{f}(x_1,\ldots,x_{i-1},0,x_{i+1},\ldots,x_n).
\]
The semigroup $\{\mathsf{P}_t\}_{t\geq0}$ of the Markovian jump process  is given by the generator $\mathsf{L}$
with transition rates $p$ from state $0$ to $1$ and $(1-p)$ from $1$ to $0$:
\[
\mathsf{L} = -\sum_{i=1}^n \mathsf{\Delta}_i.
\]
Then, we are able to derive the rate of the exponential decay in variance functions.

\begin{theo}
	[Exponential Decay of Variances for Symmetric Bernoulli Random Variables] \label{theo:Var_Ber}
	Given a Markov Triple $(\{0,1\}^n, \mathbf{\Gamma}, \mu_{n,p})$ of a Markovian jump process, one has
	\[
	\textnormal{Var}\big( \mathsf{P}_t \bm{f} \big) \leq \mathrm{e}^{-2t} \cdot \textnormal{Var}\left( \bm{f} \right),
	\]
	for any matrix-valued function $\bm{f}:\{0,1\}^n\rightarrow \mathbb{M}_d^\text{sa}$.
\end{theo}

\begin{proof}
	In Corollary \ref{coro:spectral} we show the equivalence between the exponential decay in variances and the spectral gap inequality (see Definition \ref{defn:spectral}).
	Therefore, it suffices to establish the spectral gap constant of the Markovian jump process.
	
	Notably, the spectral gap inequality of the Markov jump process is a special case of the \emph{matrix Efron-Stein inequality} in Ref.~\cite{CH1}:
	\begin{prop}
		[Matrix Efron-Stein Inequality {\cite[Theorem 4.1]{CH1}}] \label{prop:EF}
		For any measurable and bounded matrix-valued function $\bm{f}:(\mathbb{M}_d^\text{sa})^n \rightarrow \mathbb{M}_d^\text{sa}$, we have
		\begin{align} \label{eq:EF}
		\textnormal{Var} (\bm{f}) \leq \frac12 \Tr \mathbb{E} \left[ \sum_{i=1}^n \left( \bm{f}(\underline{\bm{X}}) - \bm{f}\left(\widetilde{\bm{X}}^{(i)}\right) \right)^2 \right],
		\end{align}
		where $\underline{\bm{X}} \triangleq (\bm{X}_1, \ldots, \bm{X}_n) \in (\mathbb{M}_d^\text{sa})^n$ denote an $n$-tuple random vector with independent elements, and $\widetilde{\bm{X}}^{(i)} \triangleq (\bm{X}_1,\ldots,\bm{X}_{i-1}, \bm{X}_i', \bm{X}_{i+1}, \ldots,\bm{X}_n)$ is obtained by replacing the $i$-th component of $\underline{\bm{X}}$ by an independent copy of $\bm{X}_i'$.
	\end{prop}
	By taking $\underline{\bm{X}}$ to be an $n$-tuple Bernoulli random vector, and observe that the right-hand side of Eq.~\eqref{eq:EF} coincides with $\Tr\left[ \bm{\mathcal{E}(\bm{f})} \right]$ for the Markov jump process to complete the proof.
\end{proof}

Similarly, the convergence rate of the exponential decay in entropy functionals can be calculated as follows.

\begin{theo}
	[Exponential Decay of Matrix $\Phi$-Entropies for Symmetric Bernoulli Random Variables] \label{theo:Ent_Ber}
	Given a Markov Triple $(\{0,1\}^n, \mathbf{\Gamma}, \mu_{n,p})$ of a Markovian jump process, one has
	\[
	\Ent\big( \mathsf{P}_t \bm{f} \big) \leq \mathrm{e}^{-t} \cdot \Ent\left( \bm{f} \right),
	\]
	for any matrix-valued function $\bm{f}:\{0,1\}^n\rightarrow \mathbb{M}_d^\text{sa}$.
\end{theo}
\begin{proof}
	The theorem is equivalent to proving
	\[
	\Ent\left( \bm{f} \right) \leq -  \Tr\mathbb{E} \big[\left( \bm{I} + \log\bm{f}\right)  \mathsf{L} \bm{f}  \big].
	\]
	By virtue of the subadditivity property, we first establish the case $n=1$, i.e.
	\begin{align} \label{eq:Ent_Ber2}
	\Ent\left( \bm{f} \right) \leq -  \Tr\mathbb{E} \big[\left( \bm{I} + \log\bm{f}\right)  \mathsf{L} \bm{f}  \big], \quad \forall \bm{f}:\{0,1\}\rightarrow \mathbb{M}_d^\text{sa}.
	\end{align}
	Taking $\Phi(u) = u \log u$, the first-order convexity property implies that
	\[
	\Tr\big[ \Phi(\bm{Y}) - \Phi(\bm{X}) \big] 
	\geq \Tr \big[ \mathsf{D} \Phi[\bm{X}]\left( \bm{Y}-\bm{X} \right) \big], \quad \forall \bm{X}, \bm{Y} \in\mathbb{M}_d^\text{sa}.
	\]
	Let $\bm{X}\equiv \bm{f}$ and $\bm{Y} \equiv \mathbb{E}\bm{f}$. Then it follows that
	\[
	\Tr\big[ \Phi(\mathbb{E}\bm{f}) -  \Phi(\bm{f}) \big] 
	\geq \Tr \big[  \mathsf{D} \Phi[\bm{f}]\left( \mathbb{E}\bm{f}-\bm{f} \right) \big], \quad \forall \bm{f}:\{0,1\}\rightarrow \mathbb{M}_d^\text{sa}
	\]
	from which we apply the expectation again to obtain
	\begin{align} \label{eq:Ent_Ber3}
	\Tr\big[ \mathbb{E}\Phi(\bm{f}) - \Phi(\mathbb{E}\bm{f}) \big] 
	\leq \Tr \Big[   \mathbb{E} \big[ \mathsf{D} \Phi[\bm{f}]\left( \bm{f}-\mathbb{E}\bm{f} \right) \big]    \Big].
	\end{align}
	Then by elementary manipulation, the right-hand side of Eq.~\eqref{eq:Ent_Ber3}
	leads to 
	\begin{align*}
	\Tr \Big[   \mathbb{E} \big[ \mathsf{D} \Phi[\bm{f}]\left( \bm{f}-\mathbb{E}\bm{f} \right) \big]    \Big]
	&= \Tr \big[ p(1-p) \, \mathsf{D}\Phi[\bm{f}(1)]\left(\bm{f}(1)-\bm{f}(0) \right) + (1-p)p \,  \mathsf{D}\Phi[\bm{f}(0)]\left( \bm{f}(0) - \bm{f}(1) \right)      \big]\\
	&= - \Tr \mathbb{E}\big[ \Phi'(\bm{f})\cdot \mathsf{\Delta}_1 \bm{f}  \big],
	\end{align*}
	and hence arrives at Eq.~\eqref{eq:Ent_Ber2}.
	
	Then the subadditivity of $\Phi$-entropy in Theorem \ref{theo:sub} yields
	\begin{align*}
	\Ent\left( \bm{f} \right) &\leq \sum_{i=1}^n \mathbb{E} \Big[ \Ent^{(i)} \left( \bm{f} \right) \Big] \\
	&\leq - \sum_{i=1}^n  \mathbb{E} \Big[ \Tr \mathbb{E}_i\big[ \Phi'(\bm{f})\cdot \mathsf{\Delta}_i \bm{f}  \big] \Big]\\
	&= -  \Tr\mathbb{E} \big[\left( \bm{I} + \log\bm{f}\right)  \mathsf{L} \bm{f}  \big],
	\end{align*}
	which completes the proof.
\end{proof}

\subsection{Mixing Times of Quantum Random Graphs} \label{sec:application}

In the following, we introduce a model of quantum states defined on a random graph and apply the above results to calculate the mixing time.
Consider a directed graph $\Omega$ with finite vertices.
Every arc $e=(x,y)$, $x,y\in\Omega$ of the graph corresponds a non-negative weight $w({x,y})$ (assume $y\neq x$), which represents the \emph{transition rate} starting from node $x$ to $y$.
Here we denote by $(L(x,y))_{x,y\in\Omega}$ the weight matrix that satisfies $L(x,y)\geq 0$ as $x\neq y$.
Moreover, a balance condition $\sum_{y\in\Omega} L(x,y) = 0$ for any $x\in\Omega$ is imposed.
The Markov transition kernel can be constructed via the exponentiation of the weight matrix $L$ (see e.g.~\cite{Nor97, Bak06,BGL13}):
\[
p_t(x,y) = \left( \mathrm{e}^{ t L} \right) (x,y),
\]
which stands for the probability from node $x$ to $y$ after time $t$.
Now, each vertex $x$ of the graph is endowed with a density operator $\bm{f}(x) = \bm{\rho}_x$ on some fixed Hilbert space $\mathbb{C}^d$.
The evolution of the quantum states in the graph is characterized by the Markov semigroup acting on the ensembles $\{\bm{\rho}_x\}_{x\in\Omega}$ according to the rule: 
\begin{align}
\bm{\rho}_{t,x}  \triangleq  \mathsf{P}_t \bm{f}(x) = \sum_{y\in\Omega}  \bm{\rho}_y \, p_t(x,y).
\end{align}
Thus $\bm{\rho}_{t,x}$ is the quantum state at node $x$ that is mixed from other nodes according to weight $p_t(x,y)$.

It is not hard to observe that the measure $\mu$ is invariant for this Markov semigroup $\{\mathsf{P}_t\}_{t\geq 0}$ if
\[
\mu(y) = \sum_{x\in\Omega} \mu(x) p_t (x,y), \; \forall y\in\Omega.
\]
We note that there always exists a probability measure satisfying the above equation.
However, the probability measure is unique if and only if the Markov kernel $(p_t(x,y))$ is \emph{irreducible}\footnote{
A Markov kernel matrix is called irreducible if there exists a finite $t\geq 0$ such that $p_t(x,y)>0$ for all $x$ and $y$. In other words, it is possible to get any state from any state.
	However, the uniqueness of the invariant measure gets more involved when the state space $\Omega$ is uncountable.
	We refer the interested readers to reference \cite[Chapter 7]{Pra06} for further discussions.
	We also remark that it is still unclear whether the classical characterizations of the unique invariant measures can be directly extended to the case of matrix-valued functions. This problem is left as future work.
}
\cite{Nor97, Sal97}.

As shown in Section \ref{sec:main}, all the states $\bm{\rho}_{t,x} = \mathsf{P}_t \bm{f}(x)$ will converge to the average state $\sum_{y\in\Omega} \mu(y)\cdot \bm{\rho}_y =:  \overline{\bm{\rho}}$ as $t$ goes to infinity, where $\mu$ is a unique invariant measure for $\{\mathsf{P}_t\}_{t\geq 0}$.
To measure how close it is to the average states $\overline{\bm{\rho}}$, we exploit the matrix $\Phi$-entropy functionals (with respect to the invariant measure $\mu$) to capture the convergence rate.
In particular, we choose $\Phi(u) = u^2$ and $\Phi(u) = u \log u$, which coincide the variance function and the {Holevo quantity}.

We define the $\mathbb{L}^2$ and Holevo mixing times as follows:

\begin{defn}
	Let $\bm{\rho}_t:x\mapsto \bm{\rho}_{t,x}$ be the ensembles of quantum states after time $t$.
	The $\mathbb{L}^2$ and Holevo mixing times are defined as:
	\begin{align*}
	&\tau_2(\epsilon)  \triangleq \inf\{ t: \Var(\bm{\rho}_t) \leq \epsilon\}\\
	&\tau_{\chi}(\epsilon) \triangleq \inf\{ t: \chi(\bm{\rho}_t) \leq \epsilon\}.
	\end{align*}
\end{defn}

By applying our main result (Theorem~\ref{theo:decay}), we upper bound the mixing time of the Markov random graphs. 

\begin{coro}
	Let $C_2>0$ and $C_{\chi}>0$ be the spectral gap constant and the modified log-Sobolev constant of the Markov Triple $(\Omega, \bm{\Gamma}, \mu)$.
	Then one has
	\begin{align}
	&\tau_2 (\epsilon) \leq \frac{C_2}2 \left( \log \Var(\bm{\rho}) + \log\frac1\epsilon \right) \label{eq:t_2}\\
	&\tau_{\chi} (\epsilon) \leq {C_{\chi}} \left( \log \chi(\bm{\rho}) + \log\frac1\epsilon \right), \label{eq:t_chi}
	\end{align}
	where $\bm{\rho}:x\mapsto \bm{\rho}_x$ denotes the initial ensemble of the graph.
\end{coro}
\begin{proof}
	The corollary follows immediately from Theorem \ref{theo:decay}.
	Set $\Var(\bm{\rho}_t)=\epsilon$. Then we have
	\[
	\epsilon \leq \mathrm{e}^{-2 \tau_2(\epsilon)/ C_2} \Var(\bm{\rho}),
	\]
	which implies the desired upper bound for the $\mathbb{L}^2$ mixing time. The upper bound for the mixing time of the Holevo quantity follows in a similar way.
\end{proof}

\begin{remark}
	Note that for every initial ensemble $\bm{\rho}_0$, it follows that
	\begin{align*}
	\Var(\bm{\rho}) = \Tr \left[ \sum_{x\in\Omega} \mu(x) \bm{\rho}_{x}^2 - \left( \sum_{x\in\Omega} \mu(x) \bm{\rho}_{x} \right)^2 \right] \leq \max_{x\in\Omega} \mu(x)/d =: \mu^*/d.
	\end{align*}
	Equation \eqref{eq:t_2} can be replaced by 
	\[
	\tau_2 (\epsilon) \leq \frac{C_2}2 \left( \log \frac{\mu^*}{d} + \log\frac1\epsilon \right).
	\]
	Moreover, it is well-known \cite{NC09} that the Holevo quantity $\chi(\{\mu(x), \bm{\rho}_{x}\})$ is bounded by the Shannon entropy $H(\mu)$ of the probability distribution $\mu$.
	Hence, Eq. \eqref{eq:t_chi} is rewritten as
	\begin{align*}
	\tau_2 (\epsilon) \leq \frac{C_2}2 \left( \log H(\mu) + \log\frac1\epsilon \right).
	\end{align*}	\Endremark
\end{remark}

Consider the random graph generated by the Markovian jump process on a hypercube $\{0,1\}^n$ with probability $p=1/2$. Theorems \ref{theo:Var_Ber} and \ref{theo:Ent_Ber} give the spectral gap constant $C_2=1$ and modified log-Sobolev constant $C_\chi=1$. Hence, for every initial ensemble $\bm{\rho}$, the mixing times of this quantum random graph is
\begin{align*}
&\tau_2 (\epsilon) \leq \frac{1}2 \left( \log \Var(\bm{\rho}) + \log\frac1\epsilon \right) \\
&\tau_{\chi} (\epsilon) \leq \left( \log \chi(\bm{\rho}) + \log\frac1\epsilon \right).
\end{align*}

\section{Discussions} \label{sec:diss}

Classical spectral gap inequalities and logarithmic Sobolev inequalities have proven to be a fundamental tool in analyzing Markov semigroups on real-valued functions. In this paper, we extend the definition of Markov semigroups to matrix-valued functions and investigate its equilibrium property. 
Our main result shows that the matrix $\Phi$-entropy functionals exponentially decay along the Markov semigroup, and the convergence rates are determined by the coefficients of the matrix $\Phi$-Sobolev inequality \cite{CH1}.
In particular, we establish the variance and entropy decays of the Markovian jump process using the subadditivity of matrix $\Phi$-entropies \cite{CT14} and tools from operator algebras.

The Markov semigroup introduced in this paper is not only of independent interest in mathematics, but also has substantial applications in quantum information theory. In this work, we study the dynamical process of a quantum ensemble governed by the Markov semigroups, and analyze how the entropies of the quantum ensemble evolve as time goes on.
When the quantum dynamical process is a quantum unital map, our result yields a stronger version of the monotonicity of the Holevo quantity $\chi(\{\mu(x),\mathsf{T}_t(\bm{\rho}_x)\}_x) \leq \mathrm{e}^{-t/C}\cdot\chi(\{\mu(x),\bm{\rho}_x\}_x)$. 

\section*{Acknowledgements}
MH would like to thank Matthias Christandl, Michael Kastoryano, Robert Koenig, Joel Tropp, and Andreas Winter for their useful comments. 
MH is supported by an ARC Future Fellowship under Grant FT140100574. 
MT is funded by an University of Sydney Postdoctoral Fellowship and acknowledges support from the ARC Centre of Excellence for Engineered
Quantum Systems (EQUS).

\appendix

\section{Hille-Yoshida's Theorem for Markov Semigroups} \label{sec:HY}

Classical Hille-Yoshida's theorem \cite[Chapter IX]{Yos96},  \cite[Appendix A]{BGL13} provides a nice characterization for the Dirichlet domain $\mathcal{D}(\mathsf{L})$ when the underlying Banach space $(\mathcal{B},\|\cdot\|)$ is the set of real-valued bounded continuous functions on $\Omega$.
In the following we show that the Hille-Yosida's theorem can be naturally extended to the Banach space of bounded continuous matrix-valued functions $\bm{f}:\Omega\to \mathbb{M}^\text{sa}$ equipped with the uniform norm (e.g.~$\opnorm{\bm{f}} = \sup_{x\in\Omega }\|  \bm{f}(x) \|_\infty$).
We note that the proof parallels the classical approach; see e.g.~\cite[Chapter IX]{Yos96} and \cite[Theorem 1.7]{Gui03}.

\begin{theo}[Hille-Yoshida's Theorem for Markov Semigroups] \label{theo:HY}
	A linear super-operator $\mathsf{L}$ is the infinitesimal generator of a Markov semigroup $\{\mathsf{P}_t\}_{t\geq0}$ on $\mathcal{B}$ if and only if
	\begin{itemize}
		\item The identity function belongs to the Dirichlet domain: $\bm{I} \in \mathsf{D}(\mathsf{L})$ and $\mathsf{L}\bm{I} = \bm{0}$.
		\item $\mathsf{D}(\mathsf{L})$ is dense in $\mathcal{B}$.
		\item $\mathsf{L}$ is closed.
		\item For any $\lambda>0$, $(\lambda\mathds{1} - \mathsf{L})$ is invertible. The inverse $(\lambda\mathds{1} - \mathsf{L})^{-1}$ is bounded with 
		\[
		\sup_{\opnorm{\bm{f}}\leq 1} \opnorm{ (\lambda \mathds{1} - \mathsf{L})^{-1} \bm{f} } \leq \frac{1}{\lambda}
		\]
		and preserves positivity, i.e.~$(\lambda \mathds{1} - \mathsf{L})^{-1} \bm{f} \succeq \bm{0}$ for all $\bm{f}\succeq \bm{0}$.
	\end{itemize}
\end{theo}
\begin{proof}
	~\\
\emph{Necessary condition}.

The first item follows from the fact that $\mathsf{P}_t \bm{I} = \bm{I}$ and the definition of the infinitesimal generator; see~Eq.~\eqref{eq:L}.
To prove the second item, it suffices to show that
\begin{align} \label{eq:D0}
\mathcal{D}_0 \triangleq \left\{ \frac1t \int_0^t \mathsf{P}_s \bm{f} \,\mathrm{d} s: \; \bm{f}\in\mathcal{B}, t>0 \right\} \subset \mathsf{D}(\mathsf{L}).
\end{align}
According to the continuity of the map $t\to \mathsf{P}_t \bm{f}$,  $\mathcal{D}_0$ is dense in $\mathcal{B}$ and hence completes the second item.

To show Eq.~\eqref{eq:D0}, we invoke the semigroup property in Definition~\ref{defn:Markov} to obtain 
\begin{align*}
\begin{split}
\frac1\tau (\mathsf{P}_\tau - \mathds{1}) \int_0^t \mathsf{P}_s \bm{f}\, \mathsf{d}s
&= \frac1\tau \int_{\tau}^{\tau + t} \mathsf{P}_s \bm{f}\, \mathsf{d}s - \frac1\tau \int_{0}^{t} \mathsf{P}_s \bm{f}\, \mathsf{d}s \\
&= \frac1\tau \int_{t}^{\tau + t} \mathsf{P}_s \bm{f}\, \mathsf{d}s - \frac1\tau \int_{0}^{\tau} \mathsf{P}_s \bm{f}\, \mathsf{d}s
\end{split}
\end{align*}
for any $\tau \in \mathbb{R}_+$.
By letting $\tau$ tend to zero, we have
\begin{align} \label{eq:L_int}
\mathsf{L} \int_{0}^t \mathsf{P}_s \bm{f} \, \mathrm{d} s 
= \int_{0}^t \mathsf{P}_s \mathsf{L}\bm{f} \, \mathrm{d} s
= \mathsf{P}_t \bm{f} - \bm{f},
\end{align}
for any $\bm{f} \in \mathcal{D}(\mathsf{L})$,
which shows that $\int_{0}^t \mathsf{P}_s \bm{f} \, \mathrm{d} s \in \mathcal{D}(\mathsf{L})$
and hence establishes Eq.~\eqref{eq:D0}.

To show the closeness of the generator $\mathsf{L}$, we consider a sequence $\bm{f}_n$ in $\mathcal{D}(\mathsf{L})$ converging to a function $\bm{f}$ such that
\[
\lim_{n\to \infty} \mathsf{L} \bm{f}_n = \bm{g}
\]
for some $\bm{g} \in \mathcal{B}$.
Apply Eq.~\eqref{eq:L_int} on $\bm{f}_n$ to obtain
\[
\mathsf{P}_t \bm{f}_n - \bm{f}_n = \mathsf{L} \int_0^t \mathsf{P}_s \bm{f}_n\, \mathrm{d}s = \int_{0}^{t} \mathsf{P}_s \mathsf{L} \bm{f}_n \, \mathrm{d}s.
\]
By taking $n$ go to infinity yields $\mathsf{P}_t \bm{f} -\bm{f} = \int_0^t \mathsf{P}_s \bm{g} \, \mathrm{d}s$.
Dividing by $t$ and letting $t\downarrow 0$ results in $\mathsf{L} \bm{f} = \bm{g}$ for every $\bm{f} \in \mathcal{D}(\mathsf{L})$.

To show the last point, we define the \emph{resolvent} $R(\lambda,\mathsf{L}) \triangleq (\lambda \mathds{1} - \mathsf{L})^{-1}$ of the generator $\mathsf{L}$ and show that
\begin{align} \label{eq:R}
R(\lambda,\mathsf{L}) \bm{f} = \int_0^\infty \mathrm{e}^{-\lambda s } \mathsf{P}_s \bm{f} \, \mathrm{d}s.
\end{align}
By introducing a new semigroup $\widetilde{\mathsf{P}}_s = \mathrm{e}^{-\lambda s }\mathsf{P}_s$ with associated generator $(\mathsf{L} - \lambda \mathds{1})$, Eq.~\eqref{eq:L_int} implies
\[
(\mathsf{L} - \lambda \mathds{1}) \int_0^t \mathrm{e}^{-\lambda s} \mathsf{P}_s \bm{f} \, \mathrm{d}s = \mathrm{e}^{-\lambda t } \mathsf{P}_t \bm{f} - \bm{f}.
\]
Taking the limit as $t\to \infty$ and according to the closeness of $\mathsf{L}$, we deduce that
\[
(\mathsf{L} - \lambda \mathds{1}) \int_0^\infty \mathrm{e}^{-\lambda s} \mathsf{P}_s \bm{f} \, \mathrm{d}s  - \bm{f}
\]
which establishes Eq.~\eqref{eq:R}.
Finally, since $\mathsf{P}_s$ is contractive (see Eq.~\eqref{eq:contraction}),
\begin{align} \label{eq:R2}
\sup_{\opnorm{\bm{f}}= 1} \opnorm{ R(\lambda,\mathsf{L}) } \leq \int_0^\infty \mathrm{e}^{-\lambda s} \, \mathrm{d} s = \frac1\lambda.
\end{align}
Hence $R(\lambda,\mathsf{L})$ is a bounded super-operator on $\mathcal{D}(\mathsf{L})$.
In particular, we observe the  positivity of $R(\lambda,\mathsf{L})$ from Eq.~\eqref{eq:R}.
\\~\\~
\emph{Sufficiency}.

For any $\lambda>0$, we define the Yoshida approximation $\mathsf{L}_\lambda$ of $\mathsf{L}$ by
\[
\mathsf{L}_\lambda \triangleq \mathsf{L}\lambda \mathds{1} ( \lambda \mathds{1} - \mathsf{L})^{-1} = \lambda^2 (\lambda \mathds{1} - \mathsf{L})^{-1} - \lambda \mathds{1}.
\]
From Eq.~\eqref{eq:R2}, $\mathsf{L}_\lambda$ is a bounded super-operator and hence we define a Markov semigroup
\[
\mathsf{P}_t^\lambda \triangleq \mathrm{e}^{t \mathsf{L}_\lambda} 
= \sum_{n=0}^\infty \frac{t^n}{n!} \mathsf{L}_\lambda^n = \mathrm{e}^{-t\lambda\mathds{1}} \mathrm{e}^{t\lambda^2 (\lambda \mathds{1} - \mathsf{L})^{-1}}.
\]
for $t\in\mathbb{R}_+$.
To finish the proof, it remains to show that $\mathsf{P}_t^\lambda$ converges to $\mathsf{P}_t$ as $\lambda\to \infty$ and hence the semigroup properties, contractivity, positivity, and the unit property of the semigroup $\{\mathsf{P}^\lambda_t\}_{t\geq 0}$ can be extended to $\{\mathsf{P}_t\}_{t\geq 0}$.
To achieve this, we prove that the family of super-operators $\{\mathsf{L}_\lambda\}_\lambda$ converges to $\mathsf{L}$:
\[
\lim_{\lambda\to\infty} \mathsf{L}_\lambda \bm{f} = \mathsf{L} \bm{f},
\]
and the convergence of $\{\mathsf{L}_\lambda\}_\lambda$ leads to that of the semigroups $\{\mathsf{P}_t^\lambda\}_\lambda$.
In fact, the identity 
\[
\lambda(\lambda \mathds{1} - \mathsf{L})^{-1} - (\lambda \mathds{1} - \mathsf{L})^{-1} \mathsf{L} = \mathds{1}
\]
implies that
\[
\lim_{\lambda \to \infty} \lambda (\lambda \mathds{1} - \mathsf{L} )^{-1} \bm{g} = \bm{g}
\]
for every $\bm{g} \in \mathcal{D}(\mathsf{L})$.
Hence we conclude that $\mathsf{L}_\lambda$ converges to $\mathsf{L}$ as $\lambda$ goes to infinity.
Moreover,
for any positive real numbers $(\lambda_1, \lambda_2)$ and any $\bm{f} \in \mathcal{D}(\mathsf{L})$, we have the interpolation formula
\[
\mathrm{e}^{t \mathsf{L}_{\lambda_1}} \bm{f} - \mathrm{e}^{t \mathsf{L}_{\lambda_2}} \bm{f}
= t\int_0^1 \mathrm{e}^{t \left( s\mathsf{L}_{\lambda_1} - (1-s) \mathsf{L}_{\lambda_2}\right)} 
\left( \mathsf{L}_{\lambda_1} - \mathsf{L}_{\lambda_2} \right) \bm{f} \, \mathrm{d}s.
\]
since $\mathsf{L}_{\lambda_1}$ and $\mathsf{L}_{\lambda_2}$ are commuting.
Combined with the contractivity of the semigroups $\mathrm{e}^{t \left( s\mathsf{L}_{\lambda_1} - (1-s) \mathsf{L}_{\lambda_2}\right)}$ for $s\in[0,1]$, we have
\[
\opnorm{ \mathrm{e}^{t\mathsf{L}_{\lambda_1}} \bm{f} - \mathrm{e}^{t\mathsf{L}_{\lambda_2}} \bm{f} } \leq t \opnorm{ (\mathsf{L}_{\lambda_1} - \mathsf{L}_{\lambda_2} ) \bm{f} },
\]
which shows that the convergence of $\{\mathsf{L}_\lambda\}_\lambda$ ensures the convergence of the semigroups $\{\mathsf{P}_t^\lambda\}_\lambda$.
Finally, we define  the super-operators $\mathsf{P}_t = \lim_{\lambda \to \infty} \mathsf{P}_t^\lambda$ for any $t>0$.
It is clearly that $\{\mathsf{P}_t\}_{t\geq0}$ is a Markov semigroup with generator $\mathsf{L}$ on $\mathcal{D}(\mathsf{L})$ and completes the proof.
\end{proof}

\section{Proof of Proposition \ref{prop:depolarizing}} \label{proof:Prop15}
	For convenience, we set $r=1$. Our strategy starts from proving the upper bound of $C_\chi \leq \frac12$. Then we show that the upper bound is attained when every state in the ensemble approaches the maximally mixed state $\bm{\pi}$.
	
	Firstly, recall the modified log-Sobolev constant $C_\chi$ in Eq.~\eqref{eq:Ent_const}. We assume
	\begin{align} \label{eq:1}
	\frac{\Tr\mathbb{E}_\mu[\bm{\rho}_X\log \bm{\rho}_X] + \log 2}{\Tr\mathbb{E}_\mu[\bm{\rho}_X\log \bm{\rho}_X] - \frac{\Tr\mathbb{E}_\mu[\log \bm{\rho}_X]}{2}} > \frac12,
	\end{align}
	which implies
	\begin{align} \label{eq:2}
	\Tr\mathbb{E}_\mu[\bm{\rho}_X\log \bm{\rho}_X] + 2\log 2 + \frac{\Tr\mathbb{E}_\mu[\log \bm{\rho}_X]}{2} >1.
	\end{align}
	However, since any density operator $\bm{\rho}$ on $\mathbb{C}^2$ can be expressed as 		
	\begin{align*}
	\bm{\rho} = \bm{U} \begin{pmatrix} p & 0 \\ 0 & 1-p \end{pmatrix} \bm{U}^\dagger
	\end{align*}
	for some unitary matrix $\bm{U}$ and $0\leq p\leq 1$, we have
	\begin{align*} 
	\Tr\left[ \left( \frac{\bm{I}}{2} + \bm{\rho} \right)  \log \bm{\rho}\right]= \left( \frac12 + p \right) \log p + \left( \frac12 + 1-p \right) \log (1-p) .
	\end{align*}
	The above equation is concave in $p$ and maximizes at $p=\frac12$ (i.e.~$\bm{\rho}$ is a maximally mixed state).
	As a result, we have
	\begin{align*}
	\Tr\left[ \left( \frac{\bm{I}}{2} + \bm{\rho} \right)  \log \bm{\rho}\right]  \leq -2\log 2
	\end{align*}
	which contradicts Eq.~\eqref{eq:2}. Hence, we prove the upper bound $C_\chi \leq \frac12$.
	
	Second, we consider the case of $|\Omega|=2$. Let $\mu(1) = p_1$, $\mu(2) = p_2$, where $p_1 + p_2 = 1$, and denote the Pauli matrix by
	\begin{align*}
	\bm{\sigma}_Z = \begin{pmatrix}
	1 & 0 \\ 0 & -1
	\end{pmatrix}.
	\end{align*}
	Without loss of generality, we set the states in the ensemble to be: $\bm{\rho}_1 = \bm{\pi} + \frac{\epsilon}{p_1} \bm{\sigma}_Z$ and $\bm{\rho}_2 = \bm{\pi} - \frac{\epsilon}{p_2} \bm{\sigma}_Z$, where $ 0\leq \epsilon \leq \min\{\frac{p_1}2, \frac{p_2}2 \}$.
	It is not hard to see the average state is $\bm{\pi}$.
	Then we can calculate that
	\begin{align*}
	&\frac{\Tr\mathbb{E}_\mu[\bm{\rho}_X\log \bm{\rho}_X] + \log 2}{\Tr\mathbb{E}_\mu[\bm{\rho}_X\log \bm{\rho}_X] - \frac{\Tr\mathbb{E}_\mu[\log \bm{\rho}_X]}{2}}\\
	&= \frac{ -p_1 H_b\left(\frac12 + \frac{\epsilon}{p_1}\right) - p_2 H_b\left(\frac12 + \frac{\epsilon}{p_2}\right) + \log 2}{-p_1 H_b\left(\frac12 + \frac{\epsilon}{p_1}\right) - p_2 H_b\left(\frac12 + \frac{\epsilon}{p_2}\right) + \log 2 
		- \frac{ p_1 \log\left(\frac12+ \frac{\epsilon}{p_1}\right)\left(\frac12- \frac{\epsilon}{p_1}\right) + p_2 \log\left(\frac12+ \frac{\epsilon}{p_2}\right)\left(\frac12- \frac{\epsilon}{p_2}\right)}{2}} \\
	&=: \frac{f(\epsilon)}{g(\epsilon)}=: C(\epsilon),
	\end{align*}
	where $H_b(p)\triangleq - p\log p - (1-p)\log (1-p)$ is the binary entropy function.
	By taking differentiation with respect to $\epsilon$, it can be shown that $ C''(\epsilon) <0$ and $\left. C'(\epsilon)\right|_{\epsilon = 0} =0$.
	In other words, $C(\epsilon)$ achieves its maximum when $\epsilon$ tends to $0$.
	
	Finally, by L'H\^{o}spital's rule,
	\begin{align*}
	\lim_{\epsilon \to 0} C(\epsilon) = \lim_{\epsilon \to 0} \frac{f''(\epsilon)}{g''(\epsilon)}
	= \frac{\frac{4}{p_1}+\frac{4}{p_2}}{\frac{8}{p_1}+\frac{8}{p_2}} = \frac12.
	\end{align*}
	Similarly, for the cases $|\Omega|>2$, the modified log-Sobolev constant is attained when every state in the ensemble approaches $\bm{\pi}$.
	Equation \eqref{eq:pd_Ent} then follows by Corollary \ref{coro:Ent}.


\printbibliography

\end{document}